\documentclass[a4paper,11pt]{article}

\usepackage{graphicx}
\usepackage{amsmath,amssymb}
\usepackage{algorithm} 
\usepackage{algpseudocode} 
\usepackage{color}
\usepackage{enumerate}
\usepackage{tikz}
\usepackage{url}

\newcommand{\MM}{{\mathcal M}}
\newcommand{\MS}{{\mathcal S}}
\newcommand{\MB}{{\mathcal B}}
\newcommand{\MC}{{\mathcal C}}
\newcommand{\MI}{{\mathcal I}}

\newcommand{\MO}{{\mathcal O}}

\newcommand{\bbZ}{{\mathbb{Z}}}

\newcommand{\assumref}[1]{Assumption~\ref{assum:#1}}
\renewcommand{\eqref}[1]{(\ref{eq:#1})}
\newcommand{\secref}[1]{Sect.~\ref{sec:#1}}

\newcommand{\corref}[1]{Corollary~\ref{cor:#1}}
\newcommand{\lemref}[1]{Lemma~\ref{lem:#1}}

\newcommand{\thmref}[1]{Theorem~\ref{thm:#1}}
\newtheorem{assum}{Assumption}
\newtheorem{cor}{Corollary}

\newtheorem{lem}{Lemma}

\newtheorem{thm}{Theorem}

\algrenewcommand{\algorithmicrequire}{\textbf{Input: }}
\algrenewcommand{\algorithmicensure}{\textbf{Output: }}

\long\def\invis#1{}

\newenvironment{proof}{\medskip
  \noindent{\scshape Proof:}}{\quad $\Box$\medskip}

\usepackage{amsmath,amsfonts,mathtools}
\usepackage{bm}
\usetikzlibrary{calc, positioning, arrows.meta, angles, quotes,shapes}
\newcommand{\arc}[3][ ]{\draw[#1,-{Stealth[length=3mm]}] (#2)to(#3);}

\newcommand{\basis}{B}
\newcommand{\matroid}{\mathcal M}
\newcommand{\independent}{\mathcal I}
\newcommand{\cycle}{\bm c}

\newcommand{\algfor}[2]{\For{#1} #2 \EndFor}
\newcommand{\algif}[2]{\If{#1} #2 \EndIf}

\newcommand{\cir}{\mathrm{Ex^{\rm zero}_{\rm out}}}
\newcommand{\cut}{\mathrm{Ex_{\rm in}} }

\newcommand{\set}[1]{\{ #1 \}}
\newcommand{\child}{\sol}

\newcommand{\algwhile}[2]{\While{#1} #2 \EndWhile}

\newcommand{\crank}{\mu}
\newcommand{\coefmat}{\bm B}
\newcommand{\remove}{\bm x}
\newcommand{\add}{\bm y}
\newcommand{\gft}{\mathit{GF}(2)}
\newcommand{\sets}[2]{\set{#1_{1},#1_{2},\dots,#1_{#2}}}
\newcommand{\vect}[2]{(#1_{1},#1_{2},\dots,#1_{#2})}
\newcommand{\ls}[2]{#1_{1},#1_{2},\dots,#1_{#2}}
\newcommand{\nums}[1]{1,2,\dots,#1}
\newcommand{\rtmat}{\bm X}
\newcommand{\sol}{\basis}
\newcommand{\parent}{P}
\newcommand{\rtc}{R}
\newcommand{\exc}[3]{(#1 \setminus \set{#2}) \cup \set{#3}}
\newcommand{\cast}{\cycle^\ast}

\newcommand{\successor}{\textsc{Succ}}
\newcommand{\relevant}{U^\ast}
\newcommand{\real}{\mathbb R}
\newcommand{\foundb}{\mathcal R}
\newcommand{\allmcb}{{\mathcal B}^\ast}
\newcommand{\partition}[1][I,O]{\allmcb(#1)}

\newcommand{\ang}[1]{\langle #1\rangle}

\newcommand{\lineref}[1]{line~\ref{line:#1}}
\newcommand{\rank}{r}

\newcommand{\tauMB}{\tau_{\mathrm{MinB}}}
\newcommand{\taupreREL}{\tau^{\mathrm{pre}}_{\mathrm{REL}}}
\newcommand{\tauREL}{\tau_{\mathrm{REL}}}
\newcommand{\tauMEM}{\tau_{\mathrm{IND}}}

\newcommand{\tauDET}{\tau_{\mathrm{DET}}}

\newcommand{\contract}[2]{#1 / #2}
\newcommand{\allcut}[1]{{\mathcal C}_{#1}}

\newcommand{\ledgecon}[3]{\lambda(#1,#2;#3)}
\newcommand{\procedure}[3]{\Procedure{#1}{#2} #3 \EndProcedure}
\newcommand{\maxledgecon}[3]{\lambda_{\max}(#1,#2;#3)}
\newcommand{\maxledgecong}[1]{\lambda_{\max}(#1)}

\renewcommand{\setminus}{-}
\newcommand{\MX}{{\mathcal X}}

\title{Enumeration of Bases in Matroid
  with\\ Exponentially Large Ground Set\thanks{The work is partially supported by JSPS KAKENHI Grant Number 25K14993.}}
\author{Yuki Nishimura\and Kazuya Haraguchi\thanks{Corresponding author. E-mail: {\tt haraguchi@amp.i.kyoto-u.ac.jp}}}
\date{}

\begin{document}
\maketitle

\begin{abstract}
  When we deal with a matroid ${\mathcal M}=(U,{\mathcal I})$,
  we usually assume that it is implicitly given by means of the independence (IND) oracle.
  Time complexity of many existing algorithms is polynomially bounded 
  with respect to $|U|$ and the running time of the IND-oracle.
  However, they are not efficient any more when $U$ is exponentially large in some context.
  In this paper, we propose two algorithms for enumerating matroid
  bases such that the time complexity does not depend on $|U|$. 
  For some integer $L$,
  the first algorithm
  enumerates the first $L$ minimum-weight bases 
  in incremental-polynomial time
  and the remaining ones in polynomial-delay. 
  To design the algorithm, we assume two oracles other than the IND-oracle:
  the MinB-oracle that returns a minimum basis
  and the REL-oracle that returns a relevant element one by one in non-decreasing order of weight. 
  The proposed algorithm is applicable to enumeration of minimum bases of
  binary matroids from cycle space and cut space,
  all of which have exponentially large $U$ with respect to a given graph.
  The highlight in this context is that,
  to design the REL-oracle for cut space,
  we develop the first polynomial-delay algorithm that
  enumerates all relevant cuts of a given graph in non-decreasing order of weight.
  The second algorithm
  enumerates all sets of linearly independent $r$-dimensional $r$ vectors over $\mathit{GF}(2)$ in polynomial-delay, which immediately yields
  a polynomial-delay algorithm 
  that enumerates    
  all unweighted bases of a binary matroid such that
  elements are closed under addition. 
\end{abstract}

\section{Introduction}
Let $\real,\real_{\ge0},\real_+$ denote the sets of reals,
nonnegative reals and positive reals, respectively. 
For a ground set $U$ and a family ${\mathcal I}$
of subsets of $U$,
${\mathcal M}=(U,{\mathcal I})$ is called a \emph{matroid}
if the following axioms are satisfied.
(I) $\emptyset\in{\mathcal I}$;
(II) for $I,J\subseteq U$ such that $I\subseteq J$,
$J\in{\mathcal I}$ implies $I\in{\mathcal I}$; and
(III) for $I,J\in{\mathcal I}$,
if $|I|<|J|$, then there is $x\in J\setminus I$
such that $I\cup\{x\}\in{\mathcal I}$.
A subset $I\in{\mathcal I}$
(resp., $I\not\in{\mathcal I}$) is called {\em independent}
(resp., {\em dependent}),
and a maximal independent set is called a {\em basis}.
It is well-known that the cardinality of any basis
is equal
and called the {\em rank of $\mathcal M$}.
Given a weight function $w:U\to\real$,
we write $w(S)\coloneqq\sum_{x\in S}w(x)$ for $S\subseteq U$
and call it the weight of $S$. 
A basis is {\em minimum-weight} (or {\em minimum} for short)
if the weight attains the minimum over all bases. 

Introduced by Whitney~\cite{Wh.1935}, 
matroid is among well-studied discrete structures
in theoretical computer science.
Various combinatorial problems on matroid
can be solved efficiently
(e.g., minimum-weight basis~\cite{matroidsandgreedyalgorithm},
matroid intersection~\cite{matroidintersectionandpartition,matroidintersection},
matroid partition~\cite{matroidintersectionandpartition}).
It is so expressive a mathematical model
that admits us to solve many real problems (e.g., minimum spanning tree).  
A matroid ${\mathcal M}=(U,{\mathcal I})$
is usually given by means of independence oracle
(a.k.a., membership oracle)
that returns whether $S\in{\mathcal I}$ or $S\not\in{\mathcal I}$
to a query subset $S\subseteq U$.

\emph{Enumeration problem} in general 
asks to list all required solutions without duplication
and has practical applications in
such fields as data mining and bioinformatics~\cite{AIS.1993,8845637,KK.2005,WTH.2024}.
In general, an enumeration algorithm takes
at least the computation time proportional to the output size,
where the output size can be exponentially large
with respect to the input size.
An enumeration algorithm
has been evaluated in terms of both input size and output size. 
According to \cite{JYP.1988}, the algorithm is called
\begin{itemize}
\item {\em output-polynomial}
  if the overall computation time is bounded by
  a polynomial with respect to the input size and output size;
\item {\em incremental-polynomial} ({\em incremental-poly})
  if the computation time for finding the $\ell$-th solution
  is bounded by a polynomial with respect to the input size and $\ell$; and
\item {\em polynomial-delay} ({\em poly-delay})
  if the delay (i.e., computation time between any two consecutive outputs)
  is bounded by a polynomial with respect to the input size. 
\end{itemize}

Various enumeration problems have been studied in the context of matroid~\cite{KBEGM.2005,KKW.2025,MMIB.2012}.
Uno~\cite{enummatroidbases} proposed
efficient algorithms for enumerating all bases
for such fundamental matroids as
graphic matroid,
linear matroid and
matching matroid. For graphic matroid,
there is also an algorithm for enumerating 
all bases (i.e., spanning trees)~\cite{STU.1997},
which is optimal in the sense of both time and space.
Hamacher et al.~\cite{kbestenum} presented an algorithm to enumerate
$k$-smallest bases in the sense of weight
in \( \MO(|U| \log |U| + k  \cdot |U| \cdot \rank \cdot \tauMEM) \) time,
where $\tauMEM$ denotes the computation time
of the independence oracle.
These algorithms are efficient
in the sense that the time complexity is polynomially bounded.
The bounding polynomials depend on
the cardinality $|U|$ of the ground set $U$.

In the literature, however, there are matroids
such that the ground set consists of an exponential number of elements
with respect to the size of a given graph;
e.g., binary matroids from cycle space~\cite{cycleandcutspace,mcb},
path space~\cite{minimumpathbasis}, cut space~\cite{cycleandcutspace}
and $U(T)$-space~\cite{dimofuspace}. 
These matroids are never just artificial. 
For example, a basis in the binary matroid from cycle space, called a 
{\em minimum cycle basis} in the literature, has applications in
electric engineering~\cite{electricalnetwork}, 
structural analysis of skeletal structures~\cite{skeletonstructure},
network analysis~\cite{networktheory}
and periodic event scheduling~\cite{BGV.2009}.
A minimum cycle basis can be found in polynomial time~\cite{horton,MCBfasterandsimpler,R.2020},
and there are problems of
finding other kinds of cycle basis~\cite{KLMMRUZ.2009}. 
However, 
we cannot expect the above-mentioned enumeration algorithms
to run efficiently for matroids like this. 

In this paper,
assuming a matroid with an exponentially large ground set, 
we propose two new enumeration algorithms,
one for all minimum-weight bases and
the other for all (unweighted) bases, respectively. 
The point is that, for both algorithms,
the polynomials bounding the time complexity
do not depend on the ground set size. 

After we prepare terminologies, notations and
fundamental properties in \secref{prel},
we present the main contributions of the paper
from Sections~\ref{sec:alg} to \ref{sec:allbases},
followed by concluding remarks in \secref{conc}. 


\paragraph{Enumeration of Minimum Bases (\secref{alg}).}
To design the algorithm for this purpose,
we assume two oracles other than
the independence oracle:
the minimum basis oracle that returns a minimum basis
and the relevant oracle that returns a ``relevant'' element
one by one in non-decreasing order of weight.
An element is relevant if there is a minimum basis
that contains it.
For some integer $L$,
the proposed algorithm generates the $\ell$-th solution $(\ell\le L)$
in incremental-poly time with respect to
the running times of the three oracles and the matroid rank,
and the remaining solutions are enumerated in $O(r)$ delay (i.e., poly-delay). 

\paragraph{Application to Binary Matroids (\secref{oracles}).}
We apply the minimum basis enumeration algorithm in \secref{alg}
to binary matroids from cycle space, path space
and cut space,
to show that the time complexity of
the above algorithm is polynomially bounded by
the graph size and $\ell$. 
We implement all the oracles by using previous results
as subroutines, except the relevant oracle for
the binary matroid from the cut space.
For this oracle, 
we develop the first poly-delay algorithm
that enumerates all relevant cuts in an edge-weighted graph
in non-decreasing order of weight.
This is different from a poly-delay algorithm in \cite{enumcut} that
enumerates all cuts or $s,t$-cuts (not necessarily relevant)
in non-decreasing order, while
ours concentrates on relevant cuts.

\paragraph{Enumeration of Bases (\secref{allbases}).}
We propose a poly-delay algorithm
for enumerating all bases of a binary matroid $\MM=(U,\MI)$
that satisfies the following assumption. 
\begin{assum}
  \label{assum:binary}
  For a binary matroid $\MM=(U,\MI)$,
  let $I\in\MI$ be any independent set such that $|I|\ge1$.
  For the corresponding vectors
  $\bm u_1,\bm u_2,\dots,\bm u_k\in\gft^d$
  in the elements of $I$, where $k=|I|$,
  there is an element in $U$
  that corresponds to the vector 
  $\bm u_1+\bm u_2+\dots+\bm u_k$.
\end{assum}
By this assumption, we do not require
that $U$ has an element corresponding to zero vector.
This means that elements in $U$ do not necessarily form a subspace.
We show that the delay of the proposed algorithm
is polynomially bounded with respect to
the matroid rank $r$.
The algorithm immediately yields a poly-delay algorithm
(with respect to the graph size)
for enumerating all bases in binary matroids
from the cycle space and the cut space
of a given graph.

\section{Preliminaries}
\label{sec:prel}

For disjoint subsets $X,Y$ of elements,
we may denote the disjoint union by $X\sqcup Y$
to emphasize $X\cap Y=\emptyset$.
For the minimality under the context, 
we denote by $\min(X)$ the set of minimum elements in $X$,
where we may use an abuse notation $x=\min(X)$
to indicate $x\in\min(X)$. 
For integers $p,q$ $(p\le q)$,
we denote $[p,q]\coloneqq\{p,p+1,\dots,q\}$.

Suppose that a matroid $\MM=(U,\MI)$ is given. 
We denote by $\MB$ the set of all bases of $\MM$. 
For \( \basis\in\MB\), 
we call a pair \( (x,y) \) of elements \( x \in \basis \) and \( y \notin \basis \) an \emph{exchange for \( \basis \)} 
if \( (\basis \setminus \set x) \cup \set y\in\MB \). 
We will utilize the following property of
exchange.

\begin{lem}[\cite{combination}]\label{lem:exchange}
  For a matroid ${\mathcal M}=(U,\independent)$,
  let $ \basis_1, \basis_2\in\MB$. 
  \begin{itemize}
  \item For any $ x \in \basis_1 \setminus \basis_2 $,
    there exists $ y \in \basis_2 \setminus \basis_1 $
    such that $(x,y)$ is an exchange for $\basis_1$. 
  \item For any $ y \in \basis_2 \setminus \basis_1 $,
    there exists $ x \in \basis_1 \setminus \basis_2 $
    such that $(x,y)$ is an exchange for $\basis_1$.
  \end{itemize}
\end{lem}

Suppose that a weight function $w:U\to\real$ is given.
For a basis $B\in\MB$,
let $x\in B$ be any element
and $y\notin B$ be an element such that $(x,y)$ is an exchange for $B$. 
We call \( w(x,y) \coloneqq w(y) - w(x) \) the \emph{weight of exchange \( (x,y) \)}.
We call \( (x,y) \) a \emph{zero-exchange} if \( w(x,y) = 0 \).
For \( y \in U \setminus \basis \), we define
\( \cir(\basis;y) \coloneqq \set{x \in \basis \mid (x,y) \text{ is a zero-exchange for } \basis} \).
We denote by \( \allmcb\subseteq\MB \) the set of all minimum bases
of $\mathcal M$.
An element $x\in U$ is {\em relevant}
if there is a minimum basis $B^\ast\in\allmcb$
such that $x\in B^\ast$.
We denote by $\relevant$ the set of all relevant elements.
For subsets \( I,O \subseteq U \) \( ( I \cap O = \emptyset) \), 
let us define \( \partition \) to be 
the family of all minimum bases that contain all elements in \( I \) and does not
contain any elements in \( O \),
that is, $\partition\triangleq\{B^\ast\in\allmcb\mid B^\ast\supseteq I$,
$B^\ast\cap O=\emptyset\}$. 
For \( \basis^\ast \in \partition \),
we say that
an exchange \( (x,y) \) for $\basis^\ast$
is \emph{\( I,O \)-preserving} if \( (\basis^\ast \setminus \set x) \cup \set y \)  belongs to \( \partition \).

We introduce two fundamental properties of minimum bases. 
Let $I,O\subseteq U$ be any two disjoint subsets. 
\lemref{zeroexchange} shows that, for any minimum bases
$P^\ast,Q^\ast\in\partition$, 
there is a sequence of minimum bases
$(B^\ast_1=P^\ast,B^\ast_2,\dots,B^\ast_k=Q^\ast)$
such that, for $i=1,2,\dots,k-1$,
there is  an $I,O$-preserving zero-exchange $(x_i,y_i)$ 
for $B^\ast_i$ that satisfies
$B^\ast_{i+1}=(B^\ast_i\setminus\{x_i\})\cup\{y_i\}$. 


\begin{lem}\label{lem:zeroexchange}
  Suppose that a matroid ${\mathcal M}=(U,{\mathcal I})$
  with a weight function $w:U\to\real$ is given.
  For $I,O\subseteq U$ $(I\cap O=\emptyset)$,
  let \( P^\ast, Q^\ast \in \partition \).
  For any minimum element \( y \) in \( Q^\ast \setminus P^\ast \),
  there exists an element \( x \in P^\ast \setminus Q^\ast \) such that \( (x,y) \) is an \( I,O \)-preserving zero-exchange for \( P^\ast \).
  Furthermore, $x$ is a minimum element in $P^\ast\setminus Q^\ast$. 
\end{lem}
\begin{proof}
    By Lemma~\ref{lem:exchange}, there exists an exchange \( (z,y) \) for \( P^\ast \) such that \( z \in P^\ast \setminus Q^\ast \).
    It holds that \( w(z) \leq w(y) \) since \( P^\ast \) is a minimum basis (if \( w(z) > w(y) \), then \( w(\exc{P^\ast}{z}{y}) < w(P^\ast) \) holds, which is a contradiction).
    Similarly, there is also an exchange \( (y',z) \) for \( Q^\ast \) such that \( y' \in Q^\ast \setminus P^\ast \), 
    and \( w(y') \leq w(z) \) holds.
    From the minimality of \( y \), \( w(y) \leq w(y') \) holds.
    Thus, \( w(y) = w(z) = w(y') \) holds, and hence \( w(z,y) = 0\).
    It holds that \( y \notin O \) and \( z \notin I \) since \( Q^\ast \in \partition \).
    Then  
    \( I \subseteq \exc{P^\ast}{z}{y} \) and \( O \cap (\exc{P^\ast}{z}{y}) = \emptyset \) and hence \( (x,y) \) for $x=z$ is \( I,O \)-preserving.

    Suppose that there is an element $x'\in P^\ast\setminus Q^\ast$ such that $w(x')<w(x)$.
    By the former part, there is an element $y'\in Q^\ast\setminus P^\ast$
    such that $(y',x')$ is an $I,O$-preserving zero-exchange for $Q^\ast$.
    We would have $w(y')=w(x')<w(x)=w(y)\le w(y')$, a contradiction. 
\end{proof}


%

The following \lemref{nodivide} shows a condition
that a relevant element $y\in \relevant$ is not contained
in any minimum bases in $\partition$.
In such a case, $\partition=\MB^\ast(I,O\cup\{y\})$ clearly holds.


\begin{lem}\label{lem:nodivide}
  Suppose that a matroid ${\mathcal M}=(U,{\mathcal I})$
  with a weight function $w:U\to\real$ is given.
  For $I,O\subseteq U$ $(I\cap O=\emptyset)$,
  let \( P^\ast \in \partition \) be a minimum basis
  and \( y \in \relevant\setminus P^\ast \) be a relevant element.
  If \( I \supseteq \cir(P^\ast,y) \),
  then no minimum basis in \( \partition \) contains \( y \). 
\end{lem}

\begin{proof}
  Suppose that there is a minimum basis \( Q^\ast\in \partition \)
  such that \( y \in Q^\ast \).
  By Lemma~\ref{lem:zeroexchange},
  there is a sequence of minimum bases
  $(B^\ast_1=Q^\ast,B^\ast_2,\dots,B^\ast_k=P^\ast)$
  such that, for $i=1,2,\dots,k-1$,
  there is  an $I,O$-preserving zero-exchange $(x_i,z_i)$ 
  for $B^\ast_i$ that satisfies
  $B^\ast_{i+1}=(B^\ast_i\setminus\{x_i\})\cup\{z_i\}$,
  where $z_i=\min(P^\ast\setminus B^\ast_i)$. 
  By $y\in Q^\ast\setminus P^\ast$ and 
  $|Q^\ast\setminus P^\ast|=|B^\ast_1\setminus P^\ast|>|B^\ast_2\setminus P^\ast|>\dots>|B^\ast_{k}\setminus P^\ast|=0$,
  there is $j\in\{1,2,\dots,k-1\}$ such that
  $y=x_j=\min(B^\ast_j\setminus P^\ast)$ holds by the latter part of \lemref{zeroexchange}.
  Towards such $B^\ast_j$, again by Lemma~\ref{lem:zeroexchange},
  there exists an \( I,O \)-preserving exchange
  \( (x,y) \) for \(  P^\ast \),
  where \( x \in \cir(P^\ast,y) \) and \( x \notin I \) hold.
  Then \( I \nsupseteq \cir(P^\ast,y) \) holds.
\end{proof}

\paragraph{Oracles.}
We assume that a matroid $\MM=(U,\MI)$
with a weight function $w:U\to{\mathbb R}$ is implicitly given by
IND-oracle, MinB-oracle and REL-oracle in
Sections~\ref{sec:alg} and \ref{sec:oracles}
whereas only IND-oracle is used in \secref{allbases}. 
\begin{description}
\item[Independence oracle (IND-oracle):]
  To a query subset $S\subseteq U$ of elements,
  the oracle returns whether $S$ is independent or not. 
\item[Minimum basis oracle (MinB-oracle):]
  The oracle returns a minimum basis.
\item[Relevant oracle (REL-oracle):]
  Let us denote
  the set of relevant elements by $\relevant=\{x_1,x_2,\dots,x_b\}\subseteq U$
  so that $w(x_1)\le w(x_2)\le\dots\le w(x_b)$.
  The oracle returns the $i$-th smallest relevant element $x_i$ in the $i$-th call,
  where we assume that the oracle is not called more than $b=|\relevant|$ times.
\end{description}
We denote the time complexities of these oracles by
$\tauMEM$, $\tauMB$ and $\tauREL$, respectively.
In some matroids, we may need to conduct preprocessing before the first call of the REL-oracle.
We denote by $\taupreREL$ the time complexity of the preprocessing.

\paragraph{Binary Matroids.}
We denote by $\gft$ the finite field of two elements, that is 0 and 1,
where the addition corresponds to XOR operation
and the multiplication corresponds to AND operation.
A {\em binary matroid} is a matroid $\MM=(U,\MI)$ such that
each element in $U$ is associated with 
a $d$-dimensional vector in $\gft$, where $d$ is a positive integer,
and 
$I\subseteq U$ is an independent set (i.e., $I\in\MI$)
if and only if
the vectors associated with elements in $I$ 
are linearly independent in $\gft$.

\section{An Incremental-Polynomial Algorithm for Enumerating All Minimum Bases}
\label{sec:alg}

In this section, we present an 
algorithm that enumerates all minimum bases of a matroid
$\MM=(U,\MI)$ with a weight function $w:U\to\real$,
where $\MM$ is implicitly given
by IND-oracle, MinB-oracle and REL-oracle
and we denote by $r$ the rank of $\MM$. 

In \secref{naive}, we describe how to design an incremental-poly algorithm
based on ``saturation algorithm''~\cite{MS.2019}.
We show that, in \secref{searchtree},
the time bound of the saturation algorithm can be improved
by using search tree.
Then in \secref{more}, 
we can use the search tree algorithm as a subroutine so that,
for some integer $L$,
the first $L$ solutions can be enumerated in incremental-poly time
and that the remaining solutions can be enumerated in poly-delay. 


\subsection{Na\"ive Algorithm}
\label{sec:naive}
For our problem, 
we can employ a saturation algorithm~\cite{MS.2019}
to design an incremental-poly algorithm. 
First, we generate a minimum basis $B^\ast$ by the MinB-oracle
and let $\MS\coloneqq\{B^\ast\}$.
Then we repeat the following
for each relevant element $y$ that is generated by the REL-oracle;
for each minimum basis $B$ in $\MS$,
we generate $B'\coloneqq (B-\{x\})\cup\{y\}$ for all $x\in B$ such that $(x,y)$
is a zero-exchange. 
If $B'\in\MB^\ast\setminus\MS$, then we add $B'$ to the solution set
(i.e., $\MS\coloneqq\MS\cup\{B'\}$).
A new minimum basis $B'$ will be obtained for each $y$
since it is relevant. 

Let us analyze the time complexity.
Let $\ell\coloneqq|\MS|$.
To obtain the $(\ell+1)$-st minimum basis,
we call the IND-oracle $\MO(\ell r)$ times to identify whether $B'$ is independent. 
If $B'$ is independent (i.e., a minimum basis),
we check whether it has already been generated,
where the check can be done in $\MO(r)$ time if we store $\MS$ by means of trie.
Thus the computation time to obtain the $(\ell+1)$-st solution is
\begin{align}
\MO\big((r+\ell)\tauREL+\sum_{i=1}^\ell(ir(\tauMEM+r))\big)=\MO((r+\ell)\tauREL+\ell^2r\tauMEM+\ell^2 r^2)\label{eq:naive}
\end{align}
after $\MO(\tauMB+\taupreREL)$ time preprocessing.

The space complexity is not polynomially bounded
since it needs to store all solutions generated so far. 
As mentioned in \cite{C.2023},
it is a major issue to reduce the space complexity
of enumeration algorithms of this kind.

\subsection{Search Tree Based Algorithm}
\label{sec:searchtree}
We can reduce the term $\MO(\ell^2r^2)$
in the time bound \eqref{naive}
if we realize the algorithm by using search tree.
In the search tree algorithm,
we do not need to check explicitly
whether or not a minimum basis has been already generated. 

The search tree represents a partition structure of
the set $\MB^\ast$ of minimum bases.
Each node corresponds to a subset of $\MB^\ast$
and has one minimum basis in the subset
as the representative.
For every minimum basis $B^\ast$,
there is at least one node that has $B^\ast$ as the representative. 
Traversing the search tree, 
we output the representatives
so that duplication does not occur, by which
the enumeration task is done. 

To be more precise, 
each node of the search tree is
associated with two disjoint subsets $I,O\subseteq\relevant$
of relevant elements and denoted by $\ang{I,O}$.
We call $I$ (resp., $O$) the {\em mandatory} (resp., {\em forbidden}) {\em subset} of the node $\ang{I,O}$. 
In particular, the root node
corresponds to $\ang{\emptyset,\emptyset}$. 
For $\ang{I,O}$, 
the corresponding subset of $\MB^\ast$ is $\partition$ and 
we choose a minimum basis in $\partition$ as its representative
and denote it by $R^\ast_{I,O}$. 

We branch a node in the search tree
based on the following lemma. 



\begin{lem}\label{lem:parbyexc}
  Suppose that a matroid ${\mathcal M}=(U,{\mathcal I})$ is given. 
  For disjoint subsets $I,O\subseteq\relevant$ of relevant elements,
  let $B^\ast\in\partition$ be a minimum basis and 
  $x\in B^\ast\setminus I$. 
  \begin{itemize}
  \item[\rm(i)] $\partition=\partition[I\cup\{x\},O]\sqcup\partition[I,O\cup\{x\}]$. 
  \item[\rm(ii)] Let $y\in\relevant\setminus (B^\ast\cup O)$. 
    If $(x,y)$ is an $I,O$-preserving zero-exchange for $B^\ast$, then
    $B^\ast\in\partition[I\cup\{x\},O]$ and
    $(B^\ast\setminus\{x\})\cup\{y\}\in\partition[I,O\cup\{x\}]$ hold. 
  \end{itemize}
\end{lem}

\begin{proof}
  (i)
  It is obvious that \( \partition[I \cup \set x, O] \) and \( \partition[I, O \cup \set x] \) are disjoint.
  We have \( \partition[I \cup \set x, O] \sqcup \partition[I, O \cup \set x] \subseteq \partition \)
  since \(\partition[I \cup \set x, O]\subseteq\partition \)
  and \( \partition[I, O \cup \set x]\subseteq\partition \) hold.
  Conversely, every $B^\ast\in\partition$
  satisfies either $x\in B^\ast$ or $x\notin B^\ast$.
  If $x\in B^\ast$, then $B^\ast\in\partition[I\cup\{x\},O]$ holds,
  and otherwise $B^\ast\in\partition[I,O\cup\{x\}]$ holds,
  and hence
  we have \( \partition \subseteq \partition[I \cup \set x, O] \cup \partition[I, O \cup \set x] \).

  \noindent
  (ii) It is immediate $B^\ast\in\partition[I\cup\{x\},O]$.
  By $(B^\ast\setminus\{x\})\supseteq I$ and
  $(B^\ast\cup\{y\})\cap O=\emptyset$, we see that
  $((B^\ast\setminus\{x\})\cup\{y\})\supseteq I$
  and $((B^\ast\setminus\{x\})\cup\{y\})\cap (O\cup\{x\})=\emptyset$ hold. 
\end{proof}

For a node $\ang{I,O}$ in the search tree, 
let $B^\ast\in\partition$ denote what we call the representative
and $y\in\relevant\setminus (B^\ast\cup O)$ denote what we call the {\em branching variable}.
We explain how we determine $B^\ast$ and $y$ later. 
Let \( \cir(\basis^\ast;y) \setminus I \coloneqq \{x_1,x_2,\dots,x_b\}\).
That is, for $a=1,2,\dots,b$, $(x_a,y)$ is
an $I,O$-preserving zero-exchange for $B^\ast$. 
By applying Lemma~\ref{lem:parbyexc}(i) repeatedly, 
\( \partition \) can be partitioned as follows;
\begin{align*}
  \partition &=\partition[I\cup\{x_1\},O]\sqcup\partition[I,O\cup\{x_1\}]\\
  &=\partition[I\cup\{x_1,x_2\},O]\sqcup\partition[I\cup\{x_1\},O\cup\{x_2\}]\sqcup\partition[I,O\cup\{x_1\}]\\
  &= \partition[I \cup \{x_1,x_2,\dots,x_b\},O]\sqcup\big(\bigsqcup_{a=1}^b \partition[ I \cup \{x_1,x_2,\dots,x_{a-1}\},O \cup\{x_a\}]\big).
\end{align*}
We branch the node $\ang{I,O}$
to take $b+1$ children whose mandatory and forbidden subsets appear in the above equation,
that is, $\ang{I\cup\{x_1,x_2\dots,x_b\},O},\ang{I,O\cup\{x_1\}},\dots,\ang{I\cup\{x_1,\dots,x_{b-1}\},O\cup\{x_b\}}$.
The size of the forbidden subset of the first child
(i.e., $\ang{I\cup\{x_1,x_2\dots,x_b\},O}$) is smaller than others by one. 
\invis{
To make the size of forbidden subsets equal, 
we show the following lemma.
\begin{lem}\label{lem:extendO}
  Suppose that a matroid ${\mathcal M}=(U,{\mathcal I})$
  with a weight function $w:U\to\real$ is given.
  For $I,O\subseteq U$ $(I\cap O=\emptyset)$,
  let $B^\ast\in\partition$, $y\in\relevant\setminus (B^\ast\cup O)$
  and \( \cir(\basis^\ast;y) \setminus I \coloneqq \{x_1,x_2,\dots,x_b\}\).
  It holds that \(\partition[I\cup\{x_1,x_2\dots,x_b\},O]=\partition[I\cup\{x_1,x_2\dots,x_b\},O\cup\{y\}]\).
\end{lem}

\begin{proof}
  We see that $B^\ast\supseteq I\cup\{x_1,x_2,\dots,x_b\}$
  and $B^\ast\cap O=\emptyset$, and hence $B^\ast\in\partition[I\cup\{x_1,x_2\dots,x_b\},O]$ holds.
  We have \(I \cup\{x_1,x_2,\dots,x_b\}\supseteq \cir (\basis^\ast;y)\).
  By \lemref{nodivide},
  no minimum basis in $B^\ast\in\partition[I\cup\{x_1,x_2\dots,x_b\},O]$
  contains $y$, as required. 
\end{proof}
By \lemref{extendO},
}
To make the size of forbidden subsets equal, 
observe that 
$I\cup\{x_1,x_2,\ldots,x_b\}\cup\{y\}$ contains a dependent set
by the definition of $x_1,x_2,\ldots,x_b$.
Any basis in $\mathcal B^\ast (I\cup\{x_1,x_2,\ldots,x_b\},O)$ 
does not contain $y$, and hence 
we regard $\ang{I\cup\{x_1,x_2,\dots,x_b\},O\cup\{y\}}$ as a child of $\ang{I,O}$,
instead of $\ang{I\cup\{x_1,x_2,\dots,x_b\},O}$. 

In summary, 
$\ang{I,O}$ has $\ang{I_0,O_0},\ang{I_1,O_1},\dots,\ang{I_b,O_b}$
as its children, where
\begin{align}
  I_0 &= I\cup\{x_1,x_2,\dots,x_b\};\nonumber\\
  O_0 &= O\cup\{y\};\nonumber\\
  I_a &= I\cup\{x_1,x_2,\dots,x_{a-1}\} &&\textrm{for\ }a=1,2,\dots,b;\nonumber\\
  O_a &= O\cup\{x_a\} &&\textrm{for\ }a=1,2,\dots,b.\label{eq:IO}
\end{align}
We call $\ang{I_0,O_0}$ the {\em eldest child of $\ang{I,O}$}
to distinguish $\ang{I_0,O_0}$ from the other children. 
Let $r$ denote the rank of the matroid $\MM$. 
The number of children is at least 1 when $b=0$ and
at most $r+1$ when $b=r$.
The set $\partition[I,O]$ is partitioned as follows. 
\begin{align}
  \partition = \bigsqcup_{a=0}^b \partition[I_a,O_a].
  \label{eq:part}    
\end{align}

\invis{
\begin{lem}\label{lem:sizeofo}
    If \( \partition[I',O'] \) is obtained by partitioning \( \partition \),
    then it holds that \( |O'| = |O|+1 \).
\end{lem}

\begin{proof}
    If \( O' = O \cup \set c \), then 
    \( |O \cup \set c| = |O| + 1 \) holds
    since \( c \notin O \).
    Otherwise,
    we have \( d_1,\dots,d_b \notin O \)
    since \( d_a \in \basis^\ast \) for \( a = 1,\dots,b \), 
    and hence \( |O \cup \set{d_a}| = |O| + 1 \) holds.
\end{proof}
}

We explain how to determine the representative and the branching variable 
for node $\ang{I,O}$.
As the base case, for the root $\ang{I,O}=\ang{\emptyset,\emptyset}$, 
we choose any minimum basis as its representative $R^\ast_{\emptyset,\emptyset}$.
Let \( |\relevant| \setminus R^\ast_{\emptyset,\emptyset} \coloneqq \sets y{|\relevant| - r} \),
where 
\( w(y_1) \leq w(y_2) \leq \dots \leq w(y_{|\relevant| - r}) \). 
The depth of a node is defined to be the number of edges
in the path from the root to that node.
The height of the search tree (i.e., the maximum depth)
is precisely $h_{\max}\coloneqq|\relevant| - r$.
For every node $\ang{I,O}$ in the depth $h=0,1,\dots,h_{\max}-1$,
we set the branching variable to $y_{h+1}$. 
Let $\ang{I_0,O_0},\ang{I_1,O_1},\dots,\ang{I_b,O_b}$ denote the children of $\ang{I,O}$
whose definitions are given in Eq.~(\ref{eq:IO}) for $y=y_{h+1}$. 
We define the representatives of the children as follows;
\begin{align*}
  &R^\ast_{I_0,O_0}\coloneqq R^\ast_{I,O} ;\\
  &R^\ast_{I_a,O_a}\coloneqq (R^\ast_{I,O}\setminus\{x_a\})\cup\{y_{h+1}\} &&\textrm{for\ }a=1,2,\dots,b.
\end{align*}
The point is that the eldest child $\ang{I_0,O_0}$ has the same representative
as $\ang{I,O}$, 
and that other children have different representatives from $\ang{I,O}$. 
Further, every node in the depth $h$
has a forbidden subset of the size $h$.

For $h=0,1,\dots,h_{\max}$,
let us denote by ${\mathcal N}_h$ the set of
all nodes in depth $h$. 
Lemma~\ref{lem:parbyexc} implies that
the family of subsets $\partition$ for all $\ang{I,O}\in{\mathcal N}_h$
is a partition of $\allmcb$.
Let us denote by \( \foundb_{h} \)
the set of all representatives of nodes in depth \( h \), that is,
\[
\foundb_h\triangleq\bigcup_{\ang{I,O}\in{\mathcal N}_h}\{R^\ast_{I,O}\}. 
\]

The following lemma shows that,
in every depth $h=1,2,\dots,h_{\max}$,
there is a representative that does not appear in any lower depth,
and that, for every minimum basis $B^\ast\in\allmcb$,
there is a leaf whose representative is $B^\ast$.

\begin{lem}\label{lem:newmcb}
  Suppose that a matroid ${\mathcal M}=(U,{\mathcal I})$
  with a weight function $w:U\to\real$ is given. 
  It holds that $\foundb_0\subsetneq\foundb_1\subsetneq\dots\subsetneq\foundb_{h_{\max}}=\allmcb$.
\end{lem}
\begin{proof}
  For a node $\ang{I,O}$ in depth $h$, it holds that $|O|=h$. 
  Then for each leaf $\ang{I,O}\in{\mathcal N}_{h_{\max}}$,
  $|O|=h_{\max}=|\relevant| - r$ holds
  and $\partition$ is not empty since it contains at least one minimum basis (i.e., representative). 
  A minimum basis $B^\ast\in\partition$ should satisfy
  $B^\ast\subseteq \relevant\setminus O$ and hence $|B^\ast|\le r$.
  Then $B^\ast$ is unique since it is a basis and $|B^\ast|=r$.
  We have seen that $|\partition|=1$. 
  The family of $\partition$ over all $\ang{I,O}\in{\mathcal N}_{h_{\max}}$
  is a partition of $\allmcb$,
  and hence $\foundb_{h_{\max}}=\allmcb$ holds.

  Let $h\in[0,h_{\max}-1]$.
  For every node $\ang{I,O}\in{\mathcal N}_h$,
  the representative $R^\ast_{I,O}$ is also the representative of its eldest child,
  and hence $\foundb_h\subseteq\foundb_{h+1}$ holds. 
  The element $y_{h+1}$ is relevant,
  and 
  by $\foundb_{h_{\max}}=\allmcb$, there is a leaf, say $\ang{I^\clubsuit,O^\clubsuit}$,
  whose representative $B^\ast_{I^\clubsuit,O^\clubsuit}$ contains $y_{h+1}$. 
  Let $\ang{I,O}\in{\mathcal N}_h$
  (resp., $\ang{I',O'}\in{\mathcal N}_{h+1}$) denote
  the ancestor of $\ang{I^\clubsuit,O^\clubsuit}$
  in depth $h$ (resp., $h+1$). The representative
  of $\ang{I,O}$
  does not contain $y_{h+1}$ since
  the representative of any node in ${\mathcal N}_h$ is a subset of
  $R^\ast_{\emptyset,\emptyset}\cup\{y_1,y_2,\dots,y_h\}$
  by construction of the search tree. 
  We claim that $\ang{I',O'}$ is not the eldest child of $\ang{I,O}$
  since otherwise $O'=O\cup\{y_{h+1}\}\subseteq O^\clubsuit$ would hold,
  contradicting that $y_{h+1}\in B^\ast_{I^\clubsuit,O^\clubsuit}$
  and thus $y_{h+1}\notin O^\clubsuit$.
  Then we have $y_{h+1}\in R^\ast_{I',O'}$
  and hence $R^\ast_{I',O'}\notin\foundb_h$,
  indicating that $\foundb_h\subsetneq\foundb_{h+1}$.
\end{proof}

We can enumerate all minimum bases by traversing the search tree.
Specifically, we explore the search tree in the breadth-first manner
and output minimum bases in $\foundb_0$, $\foundb_1\setminus\foundb_0$,
$\foundb_2\setminus\foundb_1$, $\dots$, $\foundb_{h_{\max}}\setminus\foundb_{h_{\max}-1}$ in this order.

The algorithm is summarized in Algorithm~\ref{alg:allMCBs}. 
In this pseudocode, the set ${\mathcal S}_h$, $h=0,1,\dots,h_{\max}$
is introduced to store all nodes in depth $h$ (i.e., ${\mathcal N}_h$).


\begin{algorithm}[t!]
  \caption{An algorithm to enumerate all minimum bases in a matroid
    where elements are weighted}
\label{alg:allMCBs}
\begin{algorithmic}[1]
  \Require{A matroid $\matroid=(U,\MI)$ with a weight function $w:U\to\real$
    that is implicitly given by
    IND-oracle, MinB-oracle and REL-oracle}
\Ensure{All minimum bases in \( \matroid \)}
\State $R^\ast_{\emptyset,\emptyset}\gets$ a minimum basis; 
\label{line:firstMCB}
\State output $R^\ast_{\emptyset,\emptyset}$;\label{line:output1}
\State ${\mathcal S}_0\gets\{\ang{\emptyset,\emptyset}\}$; \label{line:S0}
\For{$h=0,1,\dots,h_{\max}-1$} \label{line:MCB_for_begin}
\State ${\mathcal S}_{h+1}\gets\emptyset$;
\For{$\ang{I,O}\in{\mathcal S}_{h}$}\label{line:MCB_foreach_begin}
\State $\{x_1,x_2,\dots,x_b\}\gets\cir(R^\ast_{I,O};y_{h+1})\setminus I$;
\For{$a=1,2,\dots,b$}\label{line:MCB_cir} 
\State $I_a\gets I\cup\{x_1,x_2,\dots,x_{a-1}\}$; $O_a\gets O\cup\{x_a\}$;
\State $R^\ast_{I_a,O_a}\gets (R^\ast_{I,O}\setminus\{x_a\})\cup\{y_{h+1}\}$;
\State output $R^\ast_{I_a,O_a}$;\label{line:output2}
\State ${\mathcal S}_{h+1}\gets{\mathcal S}_{h+1}\cup\{\ang{I_a,O_a}\}$
\label{line:others}
\EndFor;\label{line:MCB_cir_end}
\State $I_0\gets I\cup\{x_1,x_2,\dots,x_b\}$; $O_0\gets O\cup\{y_{h+1}\}$;
\State $R^\ast_{I_0,O_0}\gets R^\ast_{I,O}$;
\State ${\mathcal S}_{h+1}\gets{\mathcal S}_{h+1}\cup\{\ang{I_0,O_0}\}$
\label{line:eldest}
\EndFor\label{line:MCB_foreach_end}
\EndFor\label{line:MCB_for_end}
\end{algorithmic}
\end{algorithm}

\begin{thm}
  \label{thm:allMCBs}
  Suppose that a matroid ${\mathcal M}=(U,\MI)$ with a weight function $w:U\to\real$
  is given. Algorithm~\ref{alg:allMCBs} enumerates all minimum bases in ${\mathcal M}$
  in incremental polynomial time with respect to the rank $r$ of ${\mathcal M}$ and the oracle running times.
  To be more precise,
  after $\MO(\tauMB+\taupreREL)$-time preprocessing, 
  the $\ell$-th minimum basis is output
  in $\MO((r+\ell)\tauREL+\ell^2 r\tauMEM)$ time. 
\end{thm}
\begin{proof}
  {\bf(Correctness)}
  We claim that ${\mathcal S}_{h}={\mathcal N}_h$, $h=0,1,\dots,h_{\max}$ holds
  upon completion of the algorithm.
  It is clearly true for $h=0$ by \lineref{S0}.
  Suppose that ${\mathcal S}_h={\mathcal N}_h$ holds as the assumption of the induction.
  For every node $\ang{I,O}\in{\mathcal S}_h$, 
  the eldest child $\ang{I_0,O_0}$ is added to ${\mathcal S}_{h+1}$
  in \lineref{eldest},
  and all other children are added to ${\mathcal S}_{h+1}$ in
  \lineref{others}, showing the claim. 
  
  No minimum basis is output more than once
  since a minimum basis $B^\ast\in{\mathcal R}_h$ is output when and only when
  it is the representative of a non-eldest child
  of its parent;
  in this case, $B^\ast\notin{\mathcal R}_0\cup{\mathcal R}_1\cup\dots\cup{\mathcal R}_{h-1}$ holds and $B^\ast$ is different from any other bases in ${\mathcal R}_h$.
  
  We show by contradiction that all minimum bases are output. 
  Suppose that there is a minimum basis that is not output by the algorithm.
  By \lemref{newmcb}, there is a leaf $\ang{I^\clubsuit,O^\clubsuit}$
  whose representative $R^\ast_{I^\clubsuit,O^\clubsuit}$ is not output.
  Then $\ang{I^\clubsuit,O^\clubsuit}$ is the eldest child of the parent,
  say $\ang{I',O'}$. 
  The representative of $\ang{I',O'}$ is the same as $\ang{I^\clubsuit,O^\clubsuit}$ (i.e., $R^\ast_{I',O'}=R^\ast_{I^\clubsuit,O^\clubsuit}$).
  The representative $R^\ast_{I',O'}$ is not output and hence $\ang{I',O'}$ is the eldest child of the parent.
  We go upward to the root in this way. 
  However, the representative of the root is output in \lineref{output1},
  a contradiction. 

  \medskip\noindent{\bf(Complexity)}
  Let $h\in\{0,1,\dots,h_{\max}-1\}$.
  By \lemref{newmcb}, there exists a node $\ang{I,O}$ in depth $h+1$
  whose representative $R^\ast_{I,O}$ is output
  when this node is visited by the algorithm. 
  Suppose that $R^\ast_{I,O}$ is the $\ell$-th output. 
  \begin{itemize}
  \item 
    The MinB-oracle is called once in \lineref{firstMCB} to construct a minimum basis as $R^\ast_{\emptyset,\emptyset}$, which can be done in $\MO(\tauMB)$ time.
  \item To use the REL-oracle, we conduct preprocessing
    in $\MO(\taupreREL)$ time. 
  \item We can determine \( y_{h+1} \) that is used
    in the for-loop from line~\ref{line:MCB_foreach_begin}
    to \ref{line:MCB_foreach_end}
    in at most \( r + h + 1 \) calls of the REL-oracle.
  \item The for-loop from line~\ref{line:MCB_foreach_begin}
    to \ref{line:MCB_foreach_end} is repeated $|{\mathcal S}_h|$ times. In each iteration,
    the most critical part is
    determination of $\{x_1,x_2,\dots,x_b\}$,
    which can be done in $r$ calls of the IND-oracle.    
  \end{itemize}
  Finally, $ \sum_{i=1}^{h} |\mathcal S_i| 
  = \sum_{i=1}^{h} |\mathcal R_i| 
  \leq h |\mathcal R_h| \le h \ell$ holds by Lemma~\ref{lem:newmcb},
  where $|{\mathcal R}_h|$ equals to the number of representatives
  that have been output by the depth $h$.
  It holds that \( h \leq \ell \) by Lemma~\ref{lem:newmcb}.
  After $\MO(\tauMB+\taupreREL)$-time preprocessing, 
  the computation time that the algorithm takes to output $R^\ast_{I,O}$
  is $\MO((r+h+1)\tauREL+\sum_{i=1}^{h+1} | \mathcal S_i | r\tauMEM)=\MO((r+\ell) \tauREL+\ell^2 r\tauMEM)$. 
\end{proof}


\subsection{A More Efficient Algorithm}
\label{sec:more}
We develop a more efficient enumeration algorithm
by using Algorithm~\ref{alg:allMCBs} as a subroutine.
We show that there is an integer $L$ such that 
the computation time to output the $\ell$-th solution is
$\MO((r+\ell)\tauREL+\ell^2r\tauMEM)$ for $\ell\le L$,
and that the delay for the remaining solutions is $O(r)$. 

Again, let $\MM=(U,\MI)$ denote the given matroid
with a weight function $w:U\to\real$. 
For a subset \( X \subseteq U \), 
we define the \emph{restriction of \( \MM \) to \( X \)}
to be the system \( \MM|X\coloneqq(X,\MI|X) \), 
where \( \MI|X \triangleq \set{I \in \MI \mid I \subseteq X} \).
The restriction $\MM|X$ is a matroid~\cite{MT.2011}.
%

Let $P^\ast=\{x_1,x_2,\dots,x_r\}$ and $Q^\ast=\{y_1,y_2,\dots,y_r\}$ 
be minimum bases of $\MM$ such that
$w(x_1)\le w(x_2)\le\dots\le w(x_r)$
and $w(y_1)\le w(y_2)\le\dots\le w(y_r)$.
By \lemref{zeroexchange},
it is easy to see that $w(x_i)=w(y_i)$ holds for all $i=1,2,\dots,r$. 
Let  $\omega_1,\omega_2,\dots,\omega_\beta\in\real$ denote the distinct weights
that appear in a minimum basis such that $\omega_1\le\omega_2\le\dots\le\omega_\beta$ and $r_1,r_2,\dots,r_\beta\in\bbZ_+$ denote their frequencies.
In other words, every minimum basis of $\MM$ contains $r_j$ elements of weight $\omega_j$,
$j=1,2,\dots,\beta$. 
We define \( \MI^\ast \triangleq \set{S \subseteq U^\ast \mid \exists \basis^\ast \in \mathcal B^\ast, S \subseteq \basis^\ast} \).
The set system $\MM^\ast\coloneqq (U^\ast,\MI^\ast)$ is a matroid~\cite{9318001}.
For $j\in[1,\beta]$, 
we also define
\[
U^\ast_j\triangleq\{e\in U^\ast\mid w(e)=\omega_j\}\ \ \ \textrm{and}\ \ \ %
\MI^\ast_j\triangleq\{S\subseteq U^\ast_j\mid S\in\MI^\ast\}.
\]
Then $\MM^\ast_j\coloneqq(U^\ast_j,\MI^\ast_j)$
is a matroid since $\MM^\ast_j$ is a restriction of $\MM^\ast$
(i.e., $\MM^\ast_j=\MM^\ast|U^\ast_j$). 

By the following lemma,
we see that the union of any bases of $\MM^\ast_1,\MM^\ast_2,\dots,\MM^\ast_\beta$
is a minimum basis of $\MM$
and that any minimum basis of $\MM$ can be partitioned into
bases of $\MM^\ast_1,\MM^\ast_2,\dots,\MM^\ast_\beta$. 
\begin{lem}
  \label{lem:decomp}
  Suppose that a matroid ${\mathcal M}=(U,{\mathcal I})$
  with a weight function $w:U\to\real$ is given.
  {\rm (i)} For $j\in[1,\beta]$, let $B^\ast_j$ be a basis of matroid $\MM^\ast_j$.
  The union $B^\ast_1\cup B^\ast_2\cup\dots\cup B^\ast_\beta$
  is a minimum basis of $\MM$.
  {\rm (ii)} For any minimum basis $B^\ast$ of $\MM$,
  $B^\ast\cap U^\ast_j$ is a basis of $\MM^\ast_j$, $j=1,2,\dots,\beta$. 
\end{lem}
\begin{proof}
  (i) 
  For $j\in[1,\beta]$, $B^\ast_j\in\MI^\ast$ holds,
  and there is a minimum basis of $\MM$,
  say $\hat{B}^\ast_j$, 
  that contains $B^\ast_j$ as a subset.
  For $k\in[1,\beta]$, let $\hat{B}^\ast_{j,k}=\hat{B}^\ast_j\cap U^\ast_k$,
  where $\hat{B}^\ast_{j,j}=B^\ast_j$ holds.
  Let us take up
  \begin{align*}
    \hat{B}^\ast_1=B^\ast_1\cup \hat{B}^\ast_{1,2}\cup\hat{B}^\ast_{1,3}\cup\dots\cup\hat{B}^\ast_{1,\beta}\ \ \ \textrm{and}\ \ \ %
    \hat{B}^\ast_2=\hat{B}^\ast_{2,1}\cup B^\ast_2\cup\hat{B}^\ast_{2,3}\cup\dots\cup\hat{B}^\ast_{2,\beta}.
  \end{align*}
  Using \lemref{zeroexchange}, we perform $(x,y)$-exchanges on $\hat{B}^\ast_2$
  such that $x\in\hat{B}^\ast_{2,1}$ and $y\in B^\ast_1$ to
  obtain a minimum basis $B^\ast_1\cup B^\ast_2\cup\hat{B}^\ast_{2,3}\cup\dots\cup\hat{B}^\ast_{2,\beta}$ of $\MM$ that contains $B^\ast_1\cup B^\ast_2$ as a subset.
  Repeating this to $\hat{B}^\ast_3,\hat{B}^\ast_4,\dots,\hat{B}^\ast_\beta$,
  we have a minimum basis $B^\ast_1\cup B^\ast_2\cup\dots\cup B^\ast_\beta$
  of $\MM$. (ii) is immediate by the definition of $\MM^\ast_j$. 
\end{proof}

Generating a minimum basis $B^\ast$ by the MinB-oracle,
we first
enumerate all bases of $\MM^\ast_1,\MM^\ast_2,\dots,\MM^\ast_\beta$. 
For $j\in[1,\beta]$, let $B_j$ be a basis of $\MM^\ast_j$
and $B^\ast_j\coloneqq B^\ast\cap U^\ast_j$.
Then we have a minimum basis $(B^\ast\setminus B^\ast_j)\cup B_j$ of $\MM$.
Once all bases of $\MM^\ast_1,\MM^\ast_2,\dots,\MM^\ast_\beta$ are obtained,
we can enumerate the remaining minimum bases of $\MM$
by taking the union of any $\beta$ bases of $\MM^\ast_1,\MM^\ast_2,\dots,\MM^\ast_\beta$,
which can be done in $O(r)$ delay. 

We can enumerate bases of a matroid $\MM^\ast_j$, $j\in[1,\beta]$
by running Algorithm~\ref{alg:allMCBs} on $\MM^\ast_j$. 
To use the algorithm,
let us confirm that the three oracles are available to $\MM^\ast_j$
so that the time complexity remains the same under the big-O notation. 
{\bf (MinB-oracle)} For the minimum basis $B^\ast$
generated by the MinB-oracle of the given $\MM$,
it suffices to take its subset of the $j$-th smallest elements,
say $B^\ast_j$. 
{\bf (IND-oracle)} A subset $B_j\subseteq U^\ast_j$
is independent for $\MM^\ast_j$ if and only if $(B^\ast\setminus B^\ast_j)\cup B_j$
is independent for $\MM$. 
{\bf (REL-oracle)}
The REL-oracle of $\MM$ lists relevant elements in $U^\ast$
in the non-decreasing order
of weight.
We can use it as the REL-oracle of $\MM^\ast_j$
as long as we deal with $\MM^\ast_1,\MM^\ast_2,\dots,\MM^\ast_\beta$
in this order. 

We have the following theorem immediately. 
\begin{thm} 
  \label{thm:allMCBs2}
  Suppose that a matroid ${\mathcal M}=(U,\MI)$ with a weight function $w:U\to\real$
  is given.
  There is an integer $L$ such that  
  we can generate the $\ell$-th minimum basis of $\MM$ in 
  $\MO((r+\ell)\tauREL+\ell^2 r\tauMEM)$ time for $\ell\le L$
  and the remaining minimum bases in $\MO(r)$ delay
  after $\MO(\tauMB+\taupreREL)$-time preprocessing. 
\end{thm}
Let $\MB^\ast_j$, $j\in[1,\beta]$
denote the set of all bases of
matroid $\MM^\ast_j$. Any integer no less than  
$1+\sum_{j=1}^\beta (|\MB^\ast_j|-1)$ suffices as $L$ in the theorem. 

\section{Application to Binary Matroids with An Exponentially Large Ground Set}
\label{sec:oracles}

In the last section,
we have presented an incremental polynomial algorithm
to enumerate all minimum bases of a given
matroid $\MM=(U,\MI)$ with a weight function $w:U\to\real$.
The computation time to output the $\ell$-th
solution is bounded by a polynomial with respect to
$\ell$, the rank $r$ of $\MM$ and
the running times of the three oracles,
and does not depend on the cardinality $|U|$ of
the ground set $U$.

In this section, we describe polynomial-time implementations
of the oracles for several binary matroids
that are well-known in the literature.
The matroids that we consider include
the binary matroids from the cycle space (\secref{oracle_cycle}),
the path space (\secref{oracle_U}),
and the cut space (\secref{oracle_cut}). 
These matroids arise from a given graph
and have a ground set of an exponential size
with respect to the graph size. 
Using our algorithm,
we can enumerate minimum bases of these matroids
efficiently.
A highlight of this section is
a poly-delay algorithm for enumerating
all relevant cuts in a given edge-weighted undirected graph
in non-decreasing order of weight,
which is used as the REL-oracle for a binary matroid
from the cut space. We explain this algorithm in \secref{oracle_cut}. 

Throughout this section, we assume that
a graph $G$ is connected and undirected. 
We denote its vertex (resp., edge) set by $V(G)$ (resp., $E(G)$),
where we let $n=|V(G)|$ and $m=|E(G)|$.
An undirected edge $\{s,t\}\in E(G)$
is written as $st$ (or $ts$) for simplicity. 
Let $ E(G) = \sets em $.
We represent a subset of edges by the incidence vector
$ \bm u = (u_1,u_2,\dots,u_m) \in \gft^m $, that is,
$ u_i = 1 $ (resp., $ 0 $) if $ e_i $ belongs (resp., does not belong) to the subset.

\subsection{Cycle Space}
\label{sec:oracle_cycle}
In this subsection, we assume an edge-weight to be positive
(i.e., $w:E(G)\to\real_+$). 
A subset $F\subseteq E(G)$ of edges
is called a \textit{cycle} if all vertex degrees in
the spanning subgraph $(V(G),F)$ are even.
The weight of a cycle \( F \) is the summation of the weight of edges in \( F \) (i.e., \( W(F) \coloneqq \sum_{e \in F} w(e) \)).
The vector space over $ \gft $ spanned by the incidence vectors of all cycles in $ G $ is called
the \textit{cycle space of $G$}.
Let $\gamma$ denote the number of connected components in $G$. 
The number $ \crank \coloneqq m-n + \gamma$ is called the \textit{cyclomatic number of $G$},
and the dimension of the cycle space is equal to $\crank$~\cite{mcb}.

We say that cycles are independent if their incidence vectors are linearly independent in $\gft$.  
Let $U$ denote the set of all cycles and
${\mathcal I}$ denote the family of all subsets of independent cycles in $U$.
We consider a matroid $\MM=(U,{\mathcal I})$.  
Obviously, the cardinality $|U|$ of the ground set can be up to
an exponential number with respect to $n$ and $m$.

\paragraph{Rank of the matroid.}
The rank of this matroid is the cyclomatic number $\crank=\MO(m)$.

\paragraph{IND-oracle.} We can decide whether a set $I\subseteq U$ of cycles
is independent or not by using Gaussian elimination on an $m \times |I|$ matrix.
We can assume that $|I|=\MO(m)$ when we run the IND-oracle in Algorithm~\ref{alg:allMCBs}. 
Then $\tauMEM=\MO(m^3)$ holds.

\paragraph{MinB-oracle.} We can construct a minimum basis in this matroid
(i.e., minimum cycle basis)
in \( \MO(m^3n) \) time by Horton's algorithm~\cite{horton}.\footnote{A faster algorithm is proposed in \cite{MCBfasterandsimpler},
  but it does not improve the preprocessing time of the enumeration algorithm in \corref{cycle}. }
We can implement the oracles so
that $\tauMB=\MO(m^3n)$ holds. 

\paragraph{REL-oracle.}
Vismara~\cite{unionofrelevant} introduced a partition of the set \( \relevant \) of relevant cycles,
say \(\relevant=U^\ast_{1}\sqcup U^\ast_{2}\sqcup\dots\sqcup U^\ast_{q}\),
such that all relevant cycles in $U^\ast_{p}$,
$p=1,2,\dots,q$ have the same weight and that $q=\MO(m^2)$,
where there is a unique cycle $x^\ast_{p}\in U^\ast_{p}$
for each $p$ that is called the {\em prototype of $U^\ast_{p}$}. 
Vismara~\cite{unionofrelevant}
presented two algorithms,
one for enumerating all
prototypes in $\MO(\mu m^3)$ time
and the other for enumerating all relevant cycles
in $U^\ast_{p}$ 
in $\MO(n|U^\ast_{p}|)$ time.
To be more precise, 
the latter algorithm attains \( \MO(n) \) delay.

We can scan all relevant cycles in the
non-decreasing order of weight as follows;
first, we enumerate all prototypes by the first algorithm
and sort them in the non-decreasing order
with respect to weight. This can be done in $\MO(\mu m^3+m^2\log m)=\MO(\mu m^3)$ time.
Let 
$w(x^\ast_{1})\le w(x^\ast_{2})\le\dots\le w(x^\ast_{q})$
without loss of generality. 
Next, for each $p=1,2,\dots,q$,
we enumerate all relevant cycles in $U^\ast_{p}$
by the second algorithm.
This can be done in $\MO(n)$ delay. 
To realize the REL-oracle, we use this enumeration algorithm as a co-routine.
We have that $\taupreREL=\MO(\mu m^3)$ and $\tauREL=\MO(n)$. 

\bigskip
By \thmref{allMCBs2}, we have the following corollary immediately. 
\begin{cor}
  \label{cor:cycle}
  Given a connected graph $G$ with
  an edge-weight function $w:E(G)\to\real_+$,
  there is an integer $L$ such that  
  we can generate the $\ell$-th minimum cycle basis $(\ell\le L)$ in 
  $\MO(\ell^2 m^4)$ time,
  where 
  $m=|E(G)|$,
  and the remaining ones in $\MO(m)$ delay
  after $\MO(m^4)$-time preprocessing.
\end{cor}

\subsection{Path Space}
\label{sec:oracle_U}
In this subsection, we assume an edge-weight to be positive
(i.e., $w:E(G)\to\real_+$). 
The concept of the path space is introduced by Hartvigsen~\cite{minimumpathbasis}.
Let \( P \subseteq V(G) \) (\( |P| \geq 2 \)) be a subset of vertices.
For two vertices \( u,v \in P \), we call a \( u,v \)-path a {\em \( P \)-path}.
The weight of a \( P \)-path \( F \) is the summation of the weight of edges in \( F \).
The vector space over \( \gft \) spanned by the incidence vectors of
all cycles and \( P \)-paths is called the {\em \( P \)-space}.\footnote{The space is called the $U$-space in the literature. To avoid confusion, we call it the $P$-space since we use $U$ to represent the ground set of a matroid in this paper.}
A set of cycles and \( P \)-paths is independent
if their incidence vectors are linearly independent in $\gft$. 
The set is 
called a {\em \( P \)-basis} if it is a basis of the \( P \)-space.

Let \( U \) denote the set of all cycles and \( P \)-paths of \( G \)
and \( \independent \) denote the family of all independent sets.
We consider a matroid \( \matroid = (U,\independent) \).
Obviously,
the cardinality \( |U| \) of the ground set can be exponential for \( n \) and \( m \).

\paragraph{Rank of the matroid.}

The rank of this matroid is \( \crank + |P| -1 = m-n+|P| = \MO(m) \)
when $G$ is connected~\cite{dimofuspace}.

\paragraph{IND-oracle.} 

We can decide whether a set \( I \subseteq U \) of cycles and \( P \)-paths is independent or not in the same way as the cycle space.
We can also assume that \( |I| = \MO(m) \) when we run the IND-oracle in Algorithm~\ref{alg:allMCBs}, and hence \( \tauMEM = \MO(m^3) \) holds.

\paragraph{MinB-oracle.} 

We can construct a minimum basis in this matroid in \( \MO(n^2m^3) \) time by using Hartvigsen's algorithm~\cite{minimumpathbasis}.
We can use this algorithm as the MinB-oracle, and \( \tauMB = \MO(n^2m^3) \) holds.

\paragraph{REL-oracle.}

Gleiss et al.~\cite{relevantpaths} generalized Vismara's algorithm~\cite{unionofrelevant} for \( P \)-space. 
The algorithm transforms the \( P \)-space of the input graph \( G \) into the cycle space of an extended graph $G'$
such that $V(G')=V(G) \cup \set{z}$
and $E(G')= E(G)\cup \bigcup_{u \in P} \set{uz}$.
Then the algorithm enumerates all prototypes in \( G' \) by applying Vismara's first algorithm, and convert them into corresponding cycles or \( P \)-paths of \( G \). 
For each converted prototype, the algorithm enumerates all relevant elements in the subset of \( U \)
to which the prototype belongs.
This algorithm is equivalent to enumerating all relevant cycles in \( G' \) by Vismara's algorithms
and converting them into corresponding cycles or \( P \)-paths of \( G \). 
Thus, we can construct the REL-oracle in the same strategy as the cycle space,
and hence \( \taupreREL = \MO((m-n+|P|)m^3)={\MO(m^4)} \) and \( \tauREL = \MO(n) \) hold.

\begin{cor}
  Given a connected graph $G$ with
  an edge-weight function $w:E(G)\to\real_+$ and
  a subset $P\subseteq V(G)$ of vertices,
  there is an integer $L$ such that  
  we can generate the $\ell$-th minimum $P$-bases in 
  \( \MO(m^4\ell^2) \) time for $\ell\le L$
  and the remaining minimum $P$-bases in $\MO(m)$ delay
  after $\MO(m^3n^2)$-time preprocessing, where $n=|V(G)|$ and $m=|E(G)|$. 
\end{cor}

\subsection{Cut Space}
\label{sec:oracle_cut}
In this subsection, we assume an edge-weight to be nonnegative
(i.e., $w:E(G)\to\real_{\ge0}$).
Recall that we assume $G$ to be connected. 
A subset $F\subseteq E(G)$ of edges is a {\em cut} ({\em of $G$})
if
$G-F$ is disconnected. 
Equivalently, 
$F$ is a cut if there is a vertex subset $W\subseteq V(G)$
such that $F=E(W;G)\coloneqq \{st\in E(G)\mid s\in W,\ t\in V(G)\setminus W\}$. 
The weight of a cut \( F \) is the total weight of edges in \( F \).
Let \( \allcut G \) denote the set of all cuts in $G$.
Let $F\in\allcut{G}$ be a cut that satisfies $F=E(W;G)$ for some $W\subseteq V$. 
For two distinct vertices \( s,t\in V \),
we call $F$ an \emph{\( s,t \)-cut}
if \( s \in W \) and \( t \notin  W \).
Let \( \allcut{G}(s,t) \) denote the set of all \( s,t \)-cuts
and \( \ledgecon{s}{t}{G} \) denote the weight of minimum \( s,t \)-cut
of \( G \).
For vertex subsets $S,T\subseteq V$, 
let \( \maxledgecon{S}{T}{G} \coloneqq \max_{a \in S,~b \in T:~a\ne b}
\ledgecon{a}{b}{G} \).
For simplicity, we write \( \maxledgecon{V}{V}{G} \)
as \( \maxledgecong{G} \).


The vector space over $\gft$ spanned by
the incidence vectors of all cuts in $G$
is called the {\em cut space of $G$}. 
We say that cuts are independent if their incidence vectors
are linearly independent in $\gft$. 
We consider a matroid $\MM=(U,\MI)$ such that 
the ground set $U=\allcut{G}$ is the set of all cuts in $G$
and $\MI$ is the family of all independent cuts in $U$. 
The cardinality $|U|$ of the ground set can be up to
an exponential number with respect to $n$ and $m$.

\paragraph{Rank of the matroid.}
It is known that the rank of this matroid
is $n-\gamma$~\cite{circuitbasis}.

\paragraph{IND-oracle.}
Similarly to previous matroids,
we can implement the IND-oracle by using Gaussian elimination
so that $\tauMEM=\MO(m^3)$ holds. 

\paragraph{MinB-oracle.}
A minimum basis in the matroid is called a {\em minimum cut basis}. 
We can construct a minimum cut basis
from a Gomory-Hu tree~\cite{gomoryhutree} of \( G \);
the set of $n-1$ edges in a Gomory-Hu tree
corresponds to a minimum cut basis~\cite{circuitbasis}.
We can construct a Gomory-Hu tree
in \( \MO(n \varphi) \) time,
where \( \varphi \) denotes the
time to compute a minimum \( s,t \)-cut in an edge-weighted undirected
graph \( G \)
for arbitrarily chosen vertices \( s,t \in V(G) \)~\cite{gomoryhutree};
e.g., $\varphi=\MO(mn)$~\cite{KRT.1994,maxflow}.
Then we can implement the MinB-oracle so that $\tauMB=\MO(n\varphi)$. 

\paragraph{REL-oracle.}
We describe how to list all relevant cuts in $G$
in non-decreasing order of weight.
This is different from the problems addressed in Yeh et al.~\cite{enumcut}. 
They studied the problem of enumerating
all cuts or $s,t$-cuts for specified vertices $s,t$
in non-decreasing order of weight
and did not take the notion of relevant cuts into account.

Recall that a cut is relevant if it belongs to
at least one minimum cut basis. 
Let \( {\mathcal C}^\ast_G\subseteq\allcut{G} \)
denote the set of all relevant cuts of \( G \).
A necessary and sufficient condition
that a cut is relevant is shown in the following lemma.

\begin{lem}[\cite{relevantcut}]\label{lem:relevantcut}
  For an edge-weighted graph $G$,
  a cut \( C \in \allcut G \) is relevant if and only if there exists a pair of distinct vertices \( s, t \in V(G) \) such that \( C \) is a minimum \( s,t \)-cut.
\end{lem}

Let us overview our algorithm. 
We build graphs $G_0\coloneqq G,G_1,\dots,G_\tau$ for some
integer $\tau$ iteratively,
where $G_{i+1}$ is obtained by contracting two certain vertices $s_i,t_i$ in $G_i$.
We will show that 
relevant cuts of $G$ appear as minimum $s_i,t_i$-cuts of $G_i$
and that the weights of minimum cuts are in non-increasing order
over $G_0,G_1,\dots,G_\tau$.
We can achieve the goal by enumerating minimum $s_i,t_i$-cuts of $G_i$
for $i=\tau,\tau-1,\dots,0$, where we employ 
Yeh et al.'s algorithm~\cite{enumcut}
to enumerate minimum $s_i,t_i$-cuts of $G_i$. 

For $X\subseteq V(G)$, 
let $\contract{G}{X}$ 
denote
the multigraph that is obtained
by contracting all vertices in \( X \) into a single vertex,
where we denote by the lowercase (i.e., $x$) the new vertex.
Multiedges may appear, and we do not take self-loops into account.  
Formally, $V(G/X)$ and $E(G/X)$ are defined as follows;
\begin{align*}
    &V(G/X) = (V(G) \setminus X) \cup \set{x},\\
    &E(G/X) = \set{ uv \in E(G) \mid u,v \notin X} \cup \left( \bigcup_{uv \in E(G):~ u \notin X,~ v \in X} \set{ux} \right). 
\end{align*}
For $v\in V(G)$, 
we regard \( \contract{G}{\{v\}}=G \).
We inherit the edge-weight function on $G/X$ from $G$
and hence regard $w$ as the edge-weight function on $G/X$. 

Let $\MX=\{X_1,X_2,\dots,X_k\}$ be a partition of $V$.
We denote by \( \contract{G}{\MX} \)
the multigraph that is obtained by contracting all vertices in $X_i\in\MX$
into a single vertex $x_i$ for each $i=1,2,\dots,k$, respectively. 
Let $G'=\contract{G}{\MX}$. 
For $1\le i<j\le k$,
every $x_i,x_j$-cut of $G'$ is an $s,t$-cut of $G$
for all $s\in X_i$ and $t\in X_j$
and hence we have $\lambda(x_i,x_j;G')\ge\maxledgecon{X_i}{X_j}{G}$.

Let $s,t\in V(G)$.
The set $\allcut{G}$ of all cuts is partitioned 
into two disjoint subsets,
one for the set of $s,t$-cuts
and the other for the set of any other cuts.
The latter subset is equal to the set $\allcut{G/\{s,t\}}$
of all cuts in the contracted graph $G/\{s,t\}$, as shown in the following lemma. 

\begin{lem}\label{lem:partitionofcuts}
    For a graph \( G \), let \( s,t \in V(G) \) be two distinct vertices.
    Then it holds that \( \allcut{G} = \allcut{G}(s,t) \sqcup \allcut{\contract{G}{\set{s,t}}} \).
\end{lem}

\begin{proof}
  Any $s,t$-cut of $G$ is not a cut in \( \contract{G}{\set{s,t}} \),
  and thus it holds that  \( \allcut{G}(s,t) \cap \allcut{\contract{G}{\set{s,t}}} = \emptyset \).
  Let \( C\in \allcut{G} \) be a cut of \( G \) such that $C=E(W;G)$
  for $W\subseteq V(G)$. 
  If \( C \) is an \( s,t \)-cut, then \( C \in \allcut{G}(s,t) \) holds.
  Otherwise, we assume that \( s,t \in V(G) \setminus W \) without loss of generality since \( E(W;G) = E(V(G)\setminus W;G) \).
  Then \( E(W;G) = E(W;\contract{G}{\set{s,t}}) \) holds, and hence \( C \in \allcut{\contract{G}{\set{s,t}}} \).
  We have seen that \( \allcut{G} \subseteq \allcut{G}(s,t) \sqcup \allcut{\contract{G}{\set{s,t}}} \) holds.  
  It is obvious that
  \( \allcut{G} \supseteq \allcut{G}(s,t) \sqcup \allcut{\contract{G}{\set{s,t}}} \) holds.
\end{proof}


For $s,t\in V(G)$, let ${\mathcal C}^\ast_G(s,t)$ denote the set of all minimum $s,t$-cuts.
By \lemref{relevantcut}, any cut in ${\mathcal C}^\ast_G(s,t)$ is relevant (i.e., ${\mathcal C}^\ast_G(s,t)\subseteq{\mathcal C}^\ast_G$). 
For a partition $\MX=\{X_1,X_2,\dots,X_k\}$ of $V(G)$,
we define
\[
\Lambda(\MX)\coloneqq \{(i,j)\mid 1\le i<j\le k,\ \ledgecon{x_i}{x_j}{G/\MX}= \maxledgecon{X_i}{X_j}{G}\}.
\]

\begin{lem}\label{lem:enoughtoberelevant}
  For an edge-weighted graph \( G \),
  let $\MX=\{X_1,X_2,\dots,X_k\}$ be a partition of $V(G)$.
  If $\Lambda(\MX)=\emptyset$, then
  no cut of $G/\MX$ is relevant for $G$.
  Otherwise, for any $(i,j)\in\Lambda(\MX)$,
  all minimum $x_i,x_j$-cuts in $G/\MX$
  are relevant for $G$. 
\end{lem}
\begin{proof}
  For the former, we show that the contraposition is true.
  Let $C\subseteq E(G/\MX)$ be a cut of $G/\MX$
  that is relevant for $G$. 
  By Lemma~\ref{lem:relevantcut},
  there are $s,t\in V(G)$ such that $C$ is a minimum $s,t$-cut in $G$. 
  Let $X_i$ and $X_j$ denote the subsets
  that contain $s$ and $t$, respectively.
  We see $i\ne j$ since $C$ separates $x_i$ and $x_j$ in $G/\MX$.
  Then we have
  \[
  w(C)\ge\lambda(x_i,x_j;G/\MX)\ge\lambda_{\max}(X_i,X_j;G)\ge\lambda(s,t;G),
  \]
  where all the inequalities hold by equalities by $w(C)=\lambda(s,t;G)$,
  showing that $(i,j)\in\Lambda(\MX)$. 
  
  For the latter,
  there exist two vertices \( s\in X_i \) and \( t \in X_j \) such that \( \ledgecon{x_i}{x_j}{G/\MX} = \maxledgecon{X_i}{X_j}{G}=\ledgecon{s}{t}{G} \).
  Any minimum \( x_i,x_j \)-cut $C$ of \( G/\MX \) is
  an \( s,t \)-cut of \( G \) that satisfies $w(C)=\ledgecon{x_i}{x_j}{G/\MX} =\ledgecon{s}{t}{G}$,
  and thus is relevant for \( G \) by Lemma~\ref{lem:relevantcut}.
\end{proof}

Let us call a cut $C$ a {\em solution} if $C\in\MC^\ast_G$. 
We outline our algorithm to enumerate all solutions in non-decreasing order of weight.
Let $V(G)=\{v_1,v_2,\dots,v_n\}$.
Initializing 
$\MX\coloneqq \{\{v_1\},\{v_2\},\dots,\{v_n\}\}$,
we execute the following: 
\begin{description}
\item[1.] If $\Lambda(\MX)\ne\emptyset$,
  then let $(i,j)\coloneqq \arg\max_{(x,y)\in\Lambda(\MX)}\{\lambda(x,y;G/\MX)\}$.
  All minimum $x_i,x_j$-cuts in $G/\MX$ are solutions by
  the latter half of \lemref{enoughtoberelevant}.
  Before enumerating them,
  we make a recursive call for the partition that is obtained
  by merging $X_i$ and $X_j$ (i.e., $(\MX\setminus\{X_i,X_j\})\cup\{X_i\cup X_j\}$)
  to enumerate solutions that have no more weights than minimum $x_i,x_j$-cuts. 
  After the recursive call is complete, we enumerate
  all minimum \( x_i,x_j \)-cuts in \( G/\MX\) as solutions, using
  the algorithm in \cite{enumcut}. 
\item[2.] Otherwise (i.e., if $\Lambda(\MX)=\emptyset$),
  we do nothing since no cut of $G/\MX$ is a solution by the
  former half of \lemref{enoughtoberelevant}. 
\end{description}
Although all outputs by the above procedure are solutions, 
it is not guaranteed that they are all solutions.
We utilize the following two lemmas to prove it. 

\begin{lem}
  \label{lem:iff_Lambda}
  For an edge-weighted graph \( G \),
  let $\MX=\{X_1,X_2,\dots,X_k\}$ be a partition of $V(G)$.
  For a cut $C$ of $G/\MX$,
  it holds that $C\in\MC^\ast_G\cap\MC^\ast_{G/\MX}$
  if and only if there is $(i,j)\in\Lambda(\MX)$ such that
  $C$ is a minimum $x_i,x_j$-cut in $G/\MX$. 
\end{lem}
\begin{proof}
  For the sufficiency, $C$ is a minimum $x_i,x_j$-cut, and by \lemref{relevantcut},
  it is relevant for $G/\MX$. By \lemref{enoughtoberelevant},
  it is also relevant for $G$ since $(i,j)\in\Lambda(\MX)$. 

  We show the necessity by contraposition.
  When $\Lambda(\MX)=\emptyset$,
  $C$ is not relevant for $G$
  (i.e., $C\notin{\mathcal C}^\ast_G$) by \lemref{enoughtoberelevant}.
  We consider the case of $\Lambda(\MX)\ne\emptyset$. 
  Let $x_i,x_j\in V(G/\MX)$ be any vertices such that 
  $C$ is an $x_i,x_j$-cut in $G/\MX$.   
  \begin{itemize}
  \item For $(i,j)\in\Lambda(\MX)$,
    $C$ is not a minimum $x_i,x_j$-cut
    by the assumption, and
    $w(C)>\lambda(x_i,x_j;G/\MX)\ge\lambda_{\max}(X_i,X_j;G)$ holds.
  \item For $(i,j)\notin\Lambda(\MX)$,
    if $C$ is a minimum $x_i,x_j$-cut, then
    $w(C)=\lambda(x_i,x_j;G/\MX)>\lambda_{\max}(X_{i},X_{j};G)$ holds,
    where the last inequality holds by $(i,j)\notin\Lambda(\MX)$.
    Otherwise (i.e., if $C$ is not a minimum $x_i,x_j$-cut),
    we have $w(C)>\lambda(x_i,x_j;G/\MX)\ge\lambda_{\max}(X_i,X_j;G)$.
  \end{itemize}
  We observe that,
  for all $x_i,x_j\in V(G/\MX)$ such that $C$ is an $x_i,x_j$-cut
  and all $(s,t)\in X_i\times X_j$, 
  $C$ is not a minimum $s,t$-cut in $G$
  since $w(C)>\lambda_{\max}(X_i,X_j;G)\ge\lambda(s,t;G)$ holds.
  We have $C\notin\MC^\ast_G$ by \lemref{relevantcut}. 
\end{proof}

\begin{lem}\label{lem:preservecut}
    For an edge-weighted graph \( G=(V,E) \), 
    let $\MX=\{X_1,X_2,\dots,X_k\}$ be a partition of $V(G)$.
    Let $(i,j)=(x,y)$ denote a maximizer of $\lambda(x,y;G/\MX)$
    among all $(x,y)\in\Lambda(\MX)$. 
    Then it holds that
    $\MC^\ast_G\cap\MC^\ast_{G/\MX}\subseteq \MC^\ast_{G/\MX}(x_i,x_j)\sqcup\MC^\ast_{(G/\MX)/\{x_i,x_j\}}$. 
\end{lem}

\begin{proof}
  We see that ${\mathcal C}^\ast_{G/\MX}(x_i,x_j)\subseteq\MC_{G/\MX}(x_i,x_j)$
  and ${\mathcal C}^\ast_{(G/\MX)/\{x_i,x_j\}}\subseteq\MC_{(G/\MX)/\{x_i,x_j\}}$
  are disjoint since $\MC_{G/\MX}(x_i,x_j)\cap\MC_{(G/\MX)/\{x_i,x_j\}}=\emptyset$
  by \lemref{partitionofcuts}.
  Let $C\in\MC^\ast_G\cap\MC^\ast_{G/\MX}$. 
  By \lemref{iff_Lambda}, there is $(i',j')\in\Lambda(\MX)$ such
  that $C$ is a minimum $x_{i'},x_{j'}$-cut.
  If $C$ is a minimum $x_i,x_j$-cut,
  then $C\in\MC^\ast_{G/\MX}(x_i,x_j)$ holds.
  Otherwise, $w(C)=\lambda(x_{i'},x_{j'},G/\MX)\le\lambda(x_i,x_j;G/\MX)$
  holds by the definition of $(i,j)$,
  and hence $C$ is not an $x_i,x_j$-cut. 
  Let $W\subseteq V(G/\MX)$
  denote the vertex subset such that $C=E(W;G/\MX)$.
  We assume $x_i,x_j\in V(G/\MX)\setminus W$ without loss of generality
  and let $x$ denote the vertex in $(G/\MX)/\{x_i,x_j\}$
  that is obtained by contracting $x_i$ and $x_j$. 

  Let $G'\coloneqq (G/\MX)/\{x_i,x_j\}$ for convenience. 
  By Lemma~\ref{lem:partitionofcuts},
  \( C \) is a cut of \(G'\). 
  For contradiction,
  suppose that \( C \) is not relevant for \(G'\)
  (i.e., $C\notin {\mathcal C}^\ast_{G'}$).
  By Lemma~\ref{lem:relevantcut},
  for any two vertices $u \in W$,
  and $v\in V(G')\setminus W$,  
  there is a \( u,v \)-cut $C'$ of \(G'\)
  such that $w(C')<w(C)$. 
  If \( v = x \), then $C'$ is a \( u,x_i \)-cut as well as
  a \( u,x_j \)-cut of \(G/\MX\),
  and otherwise, $C'$ is a \( u,v \)-cut of \( G/\MX \).
  Then $C'$ is a $(u',v')$-cut of $G/\MX$
  for any two vertices \( u' \in W, v' \in V(G/\MX) \setminus W \),
  and $C$ would not be relevant for \( G/\MX \) by Lemma~\ref{lem:relevantcut},
  a contradiction. 
\end{proof}

We show the enumeration algorithm in Algorithm~\ref{alg:enumallrelevantcut}.
In {\bf 1} of the above outline, to compute $\Lambda(\MX)$,
we identify whether $\lambda(x_i,x_j;G/\MX)=\lambda_{\max}(X_i,X_j;G)$ holds
for all $x_i,x_j\in V(G/\MX)$. 
During the algorithm, we store
$\lambda_{\max}(X_i,X_j;G)$ by using an array $W[X_i,X_j]$,
where we assume that the key is unordered;
i.e., $W[X_i,X_j]=W[X_j,X_i]$.
First, we initialize $W[\{u\},\{v\}]\coloneqq \lambda(u,v;G)$ for any two $u,v\in V(G)$.
When we contract $x_i,x_j\in G/\MX$ into a new vertex,
for all $x_r\in V(G/\MX)\setminus\{x_i,x_j\}$,
we set $\max\{W[X_r,X_i],W[X_r,X_j]\}$ to the new entry $W[X_r,X_i\cup X_j]$
since $\lambda_{\max}(X_r,X_i\cup X_j;G)=\max\{\lambda_{\max}(X_r,X_i;G),\lambda_{\max}(X_r,X_j;G)\}$. 

\begin{algorithm}[t!]
\caption{An algorithm to enumerate all relevant cuts of a given edge-weighted connected undirected graph \( G \)}\label{alg:enumallrelevantcut}
\begin{algorithmic}[1]
\Require{An edge-weighted connected undirected graph \( G \)}
\Ensure{All relevant cuts of \( G \) in non-decreasing order of weights}
\State $\MX\gets\{\{v_1\},\{v_2\},\dots,\{v_n\}\}$;
\Comment{$V(G)=\{v_1,v_2,\dots,v_n\}$}
\algfor{each  \( u,v \in V(G) \) $(u\ne v)$\label{line:initw}}{
  \State \( W[\{u\},\{v\}] \gets \ledgecon{u}{v}{G} \) 
}\label{line:endinitw};
\State call \textsc{EnumRelevantCuts}(\( \MX\))
\Statex

\procedure{EnumRelevantCuts}{$\MX=\{X_1,X_2,\dots,X_k\}$}{ \label{line:beginrecursivecall}
  \State Compute $\Lambda(\MX)$; \label{line:Lambda}
  \algif{\( \Lambda(\MX)\ne\emptyset  \) \label{line:numofvertices}}{
    \State $(i,j)\gets\arg\max_{(x,y)\in\Lambda(\MX)}\{W[X_x,X_y]\}$;
    \label{line:chooseij}
    \algfor{each \( X_r\in\MX \setminus \{X_i,X_j\} \) \label{line:updatew1}}{
      \State \( W[X_r,X_i\cup X_j] \gets \max\{W[X_r,X_i],W[X_r,X_j]\} \)
    }\label{line:endupdatew1};
    \State call \textsc{EnumRelevantCuts}(\( (\MX\setminus\{X_i,X_j\})\cup\{X_i\cup X_j\} \)); \label{line:recursivecall}
    \State output all minimum \( x_i,x_j \)-cuts of \( G/\MX \) \label{line:output}
  };
}\label{line:endrecursivecall}
\end{algorithmic}
\end{algorithm}

\begin{thm}
  \label{thm:allcuts}
  Given a connected undirected graph \( G \) with an edge-weight function \( w: E(G) \to \real_{\geq 0} \), 
  Algorithm~\ref{alg:enumallrelevantcut} outputs all relevant cuts of \( G \) in non-decreasing order of weight 
  in \( \MO(nm\log(n^2/m)) \) delay except that
  the computation time for the first output is $\MO(n^2\varphi)$.
\end{thm}

\begin{proof}
  {\bf(Correctness)}
  The algorithm halts in a finite time
  since $|\MX|$ decreases by one every time when
  \textsc{EnumRelevantCuts} is called recursively
  and $\Lambda(\MX)=\emptyset$ holds for $|\MX|=1$.
  Suppose that \textsc{EnumRelevantCuts} is called $T$ times
  and that the $t$-th call is invoked with partition $\MX_t$ of $V(G)$, 
  where $\MX_1=\{\{v_1\},\{v_2\},\dots,\{v_n\}\}$
  and $\MX_{t+1}=(\MX_t\setminus\{X_{i_t},X_{j_t}\})\cup\{X_{i_t}\cup X_{j_t}\}$ for some $X_{i_t},X_{j_t}\in\MX_t$,
  $t=1,2,\dots,T-1$. 

  Let $\MX$ be a partition of $V(G)$. 
  For $(i,j)\in\Lambda(\MX)$,
  $\MC^\ast_{G/\MX}(x_i,x_j)\subseteq {\mathcal C}^\ast_G$
  holds by the latter part of \lemref{enoughtoberelevant}.
  Then by \lemref{preservecut}, we have 
\begin{align*}
  \MC^\ast_G\cap \MC_{G/\MX}
  =\MC^\ast_G\cap(\MC^\ast_G\cap \MC_{G/\MX})
  \subseteq
  \MC^\ast_{G/\MX}(x_i,x_j)\sqcup (\MC^\ast_G\cap\MC^\ast_{(G/\MX)/\{x_i,x_j\}}),
\end{align*}
where we can apply \lemref{preservecut} recursively
to the last term $\MC^\ast_G\cap\MC^\ast_{(G/\MX)/\{x_i,x_j\}}$.
Using $\MX_1,\MX_2,\dots,\MX_T$, we have
\begin{align*}
  \MC^\ast_G=\MC^\ast_{G}\cap\MC^\ast_{G/\MX_1}
  &\subseteq \MC^\ast_{G/\MX_1}(x_{i_1},x_{j_1})\sqcup (\MC^\ast_G\cap\MC^\ast_{G/\MX_2})\\
  &\subseteq \bigsqcup_{t=1}^{T-1}(\MC^\ast_{G/\MX_t}(x_{i_t},x_{j_t}))\sqcup (\MC^\ast_G\cap\MC^\ast_{G/\MX_T}), 
\end{align*}
where the last term $\MC^\ast_G\cap\MC^\ast_{G/\MX_T}$ is empty
by the former part of \lemref{enoughtoberelevant}
since $\Lambda(\MX_T)=\emptyset$. 
Again, $\MC^\ast_{G/\MX}(x_{i_t},x_{j_t})\subseteq {\mathcal C}^\ast_G$ holds,
and hence we have
\[
\MC^\ast_G=\bigsqcup_{t=1}^{T-1}(\MC^\ast_{G/\MX_t}(x_{i_t},x_{j_t})). 
\]
For $t=1,2,\dots,T-1$,
the solutions in $\MC^\ast_{G/\MX_t}(x_{i_t},x_{j_t})$
are maximum relevant cuts for $G$
among those in $G/\MX_t$. 
Solutions in the subset $\MC^\ast_{G/\MX_t}(x_{i_t},x_{j_t})$
are output after those
in $\MC^\ast_{G/\MX_s}(x_{i_s},x_{j_s})$, $s=t+1,\dots,T-1$
are output; see lines \ref{line:recursivecall} and \ref{line:output}.
This shows that all solutions are output in
non-decreasing order of weight. 

\medskip\noindent{\bf(Complexity)}
By using a Gomory-Hu tree,
we can initialize the array \( W \) in
\( \MO(n\varphi+n^2) \) time
(i.e., lines~\ref{line:initw} to \ref{line:endinitw}). 

In each recursive call of \textsc{EnumRelevantCuts}, 
We can compute $\Lambda(\MX)$ in
\( \MO(n\varphi+n^2) \) time (\lineref{Lambda}),
where the contracted graph $G/\MX$ can be generated in $\MO(m)$ time.
Lines~\ref{line:chooseij} to \ref{line:endupdatew1}
can be done in $\MO(n^2)$ time.
The first solution is output after the $T$-th recursive call,
where $T=\MO(n)$.
According to \cite{enumcut},
for two distinct vertices \( x_i,x_j \in V(G/\MX) \), 
we can enumerate all \( x_i,x_j \)-cuts of
\( G/\MX \) (not necessarily relevant)
in non-decreasing order of weight
in \( \MO(nm\log(n^2/m)) \) delay. 
Using this algorithm, 
we can obtain
all minimum \( x_i,x_j \)-cuts in \( \MO(nm\log(n^2/m)) \) delay in line~\ref{line:output}.
The computation time for the first output is
$\MO(n(n\varphi+n^2))=\MO(n^2\varphi)$. 
\end{proof}

Regarding that the computation until the first output
as preprocessing of the REL-oracle,
we have \( \taupreREL = \MO(n^2\varphi) \) and \( \tauREL = \MO(nm\log(n^2/m)) \).


\bigskip
\begin{cor}
  Given a connected graph $G=(V,E)$ with
  an edge-weight function $w:E\to\real_{\geq 0}$,
  there is an integer $L$ such that  
  we can generate the $\ell$-th minimum cut basis in
  \( \MO((n + \ell) nm \log (n^2/m) + \ell^2 nm^3) \) time
  for $\ell\le L$
  and the remaining minimum cut bases in $\MO(n)$ delay
  after $\MO(n^2\varphi)$-time preprocessing,
  where $n=|V|$ and  $m=|E|$.
\end{cor}

\section{A Polynomial-Delay Algorithm for Enumerating All Bases}
\label{sec:allbases}

In this section,
we propose a poly-delay algorithm
for enumerating all bases of a binary matroid $\MM=(U,\MI)$
that satisfies \assumref{binary}. 
%
For convenience, we may represent a vector (resp., matrix) over $\gft$
by a binary vector (resp., matrix) throughout the section.
Such notions as 
linear independence of binary vectors
and the determinant of a binary matrix are also considered over $\gft$. 
We assume only IND-oracle,
which can be realized by any algorithm that computes a matrix determinant
over $\gft$.
Note that no other oracle is assumed like previous sections.

The algorithm is based on Avis and Fukuda's well-known
{\em reverse search}~\cite{reversesearch}. 
We show that the delay of the algorithm is
$\MO(r^3+r^2\tauDET(r))$, where
$r$ denotes the rank of $\MM$ and 
$\tauDET(r)$ denotes the computation time
for obtaining the determinant of a given $r\times r$ binary matrix. 
Again, the polynomial does not depend on the ground set size $|U|$.
The algorithm immediately yields
poly-delay enumeration of bases in binary matroids
from cycle space and cut space. 

Let us remark that the vector space
has already been used in enumeration research.
For example, in \cite{cyclespacemethod,allcircuits}, 
enumeration of cycles in graphs is studied
under the cycle space. 

\subsection{Solution Space}
In a binary matroid $\MM$ with rank $r$, 
each element is associated with a vector in $\gft^d$,
where $d$ denotes the dimension.
For example, in the case of cycle/cut space of a graph $G$,
we have $d=|E(G)|$ and
each value in the vector indicates whether the corresponding edge
belongs to the subset (1) or not (0).

A basis of $\MM$ corresponds to a set of $d$-dimensional $r$ vectors.
Instead of $d$-dimensional $r$ vectors, 
we regard a set of $r$-dimensional $r$ vectors
that are obtained in the following way as a solution.
By this, we can exploit such notions from linear algebra
as cofactor expansion and determinant. 

Let $X=\{\bm x_1 ,\bm x_2 , \dots, \bm x_r \}$
be a set of $d$-dimensional $r$ vectors of a basis in $\MM$.
Regarding each vector as a row vector,
we represent its
$r\times d$ matrix representation by the boldface $\bm X$.
For any binary vector $\bm c\in\gft^r$,
there is an element in $U$
that corresponds to vector $\bm u = \bm c\bm X\in\gft^d$ 
by \assumref{binary}.

Let  $B=\{\bm c_1,\bm c_2,\dots,\bm c_r\}$
be a set of $r$-dimensional $r$ vectors.
By the following lemma,  
$\bm U=\bm B\bm X$ is the matrix representation of
the incidence vectors of a basis 
if and only if
$\bm c_1,\bm c_2,\dots,\bm c_r$ are linearly independent. 
%

\invis{
It is known that the sum of cycles (resp., cuts) is also a cycle (resp., a cut) 
if the sum is not \( \bm 0 \)~\cite{cycleandcutspace}.
Thus, the linear combination of the elements in a basis is also an element in \( U \).
Assume that a basis \( \rtc = \sets{c}{\rank} \) is given.
For $ i = \nums{\rank} $, let $ \bm x_i \in \gft^m $ be the incidence vector of the element $ c_i \in \rtc $ 
and $ X \coloneqq \sets{\bm x}\rank $.
For any incidence vector $ \bm u \in \gft^m $, 
there exists a unique $ \rank $-dimensional vector 
$ \cycle \in \gft^\rank $ 
such that $ \bm u = \cycle \rtmat $.
}

\begin{lem}\label{lem:independence}
  Let $\MM=(U,\MI)$ be a binary matroid that satisfies \assumref{binary}
  and $X=\{\bm x_1 ,\bm x_2 , \dots, \bm x_r \}$ be a set of
  $d$-dimensional vectors that correspond to a basis in $\MM$. 
  For $ \cycle_1,\cycle_2,\dots,\cycle_k \in \gft^\rank $,
  let $ \bm u_1 = \cycle_1 \rtmat, \bm u_2 = \cycle_2 \rtmat,\dots, \bm u_k = \cycle_k \rtmat$. 
  Then  $ \cycle_1,\cycle_2,\dots,\cycle_k $ are linearly independent if and only if $ \bm u_1,\bm u_2,\dots,\bm u_k $ are linearly independent.
\end{lem}
\begin{proof}
    The linear mapping $ f: \gft^{\rank} \to \gft^{m}$; $ \cycle \mapsto \cycle \rtmat $ is injective.
    The vectors $ \bm u_1,\bm u_2,\dots,\bm u_k $ are independent if and only if $ \cycle_1, \cycle_2, \dots, \cycle_k $ are independent.
\end{proof}

Choose an arbitrary basis of $\MM$
and let $X=\{\bm x_1 ,\bm x_2 , \dots, \bm x_r \}$
denote the set of the vectors corresponding to the basis.
We see that our task is to enumerate all sets
$B=\{\bm c_1,\bm c_2,\dots,\bm c_r\}$ of
$r$-dimensional binary vectors
that are linearly independent,
by which the matrix representations $\bm U=\bm B\bm X$
of the incidence vectors of all bases of $\MM$ are enumerated. 
For convenience, we may call such $B$ a basis
and an $r$-dimensional binary vector 
an element. 
We denote elements $(1,1,\dots,1)$ and $(0,0,\dots,0)$
by $\bm 1$ and $\bm 0$, respectively.

\subsection{Reverse Search Algorithm}

\invis{
In this section, we present a poly-delay algorithm that enumerates all bases of cycle/cut space.
This algorithm is based on reverse search~\cite{reversesearch}, 
where we define the parent based on the exchange for a basis.

Let \( G = (V,E) \) be a graph and 
\( \matroid = (U, \independent) \) be a binary matroid, 
where \( U \) is the set of incidence vectors of all cycles/cuts of \( G \).

In the following, 
we first show a manner to construct a basis, 
which is used as a root of the reverse search algorithm, 
in Subsection~\ref{sss:root}.
We show that the task of enumeration can be reduced to the enumeration of the bases of \( \gft^\rank \) in Subsection~\ref{sss:reduction}.
Then we define the parent-child relationship in Subsection~\ref{sss:parent}.
We show an algorithm to enumerate all children of a basis in Subsection~\ref{sss:children} and show a complete reverse search algorithm to enumerate all basis in Subsection~\ref{sss:allbases}.
}
We exploit reverse search~\cite{reversesearch}
to design the algorithm. We do the following tasks:
\begin{itemize}
\item select a solution as the root;
\item define the parent of each non-root solution; and
\item design an algorithm that generates all children
  of a given solution, 
\end{itemize}
where a solution in this case 
is a basis of the matroid $\MM$.
Traversing the family-tree from the root,
we can enumerate all solutions.

\subsubsection{Definition of Root}
For the root of the reverse search,
we take a basis $R$ whose matrix representation is a basis matrix.
To be more precise, we set $R=\{\cycle^\ast_1, \cycle^\ast_2, \dots,\cycle^\ast_r \}$ as the root, 
where we denote by $c^\ast_i$, $i=1,2,\dots,r$ an
$r$-dimensional binary vector such that
the $i$-th entry takes 1 and all others take 0.

\invis{
  By \lemref{independence}, the task of enumerating all bases of $ \matroid $ is reduced to enumeration of all sets of independent vectors $ \ls{\cycle}{\rank} \in \gft^\rank $, 
  that is the all bases of \( \gft^\rank \).
  For convenience, unless no confusion arises, 
  we may call $ \cycle \in \gft^\rank $ an element  
  and call a subset $ \sets{\cycle}{\rank} \subseteq \gft^\rank $ a basis 
  if \( \sets{\cycle}{\rank} \) is linearly independent.
  For a basis \( \basis \), we denote by the boldface \( \coefmat \) the \( \rank \times \rank \) matrix consisting of the \( \rank \)-dimensional vectors corresponding to the elements in \( \basis \).
  Let $ \rtc = \sets{\cycle^\ast}{\rank} $ denote the elements in $ \rtc $.
  By the definition of $ \cycle^\ast $, it holds that $ \cycle^\ast_i = \bm e_i $, where $ \bm e_i $ is the vector whose $ i $-th entry is $ 1 $ and the others are $ 0 $.
}

Given a set $S$ of $r$ elements,
we can decide that $S$ is (resp., is not) a basis
if $\det\bm S=1$ (resp., 0). 
Let $ \sol = \sets{\cycle}{\rank} $ be a basis.
For the $ \rank \times \rank $ matrix $ \coefmat $ and $ i,j \in [1,\rank] $, we denote by $ \coefmat^{(i,j)} $ the $ (\rank-1) \times (\rank-1) $ matrix that is obtained by deleting the $ i $-th row and $ j $-th column from $ \coefmat $.
In particular, for an element $ \cycle= \cycle_i $,  we may write $ \coefmat^{(i,j)} $ by $ \coefmat^{(\cycle,j)} $. 

\begin{lem}\label{lem:cofactorexpansion}
  Let $ \sol\subseteq\gft^{r} $ be a basis
  and $ \cycle = \vect{c}{\rank} \in \sol $ be an element.
  It holds that $ \sum_{j=1}^{\rank} c_j \det (\coefmat^{(\cycle,j)})  = 1 $.
\end{lem}
\begin{proof}
  Elements in a basis are linearly independent
  and hence $ \det \coefmat = 1 $ holds.
  By the cofactor expansion, 
  it follows that 
  $ \det \coefmat = \sum_{j=1}^{\rank} c_j \det (\coefmat^{(\cycle,j)}) = 1 $.
\end{proof}

\subsubsection{Definition of Parent}
For an element $\bm c=(c_1,c_2,\dots,c_r)\in\gft^\rank$,
we denote by $\textrm{bin}(\bm c)$ the number
$(c_1c_2\dots c_r)_2$ in the binary system. 
For two elements $\bm c,\bm d\in\gft^\rank$,
we write $\bm c\prec\bm d$ (or $\bm d\succ\bm c$)
if $\textrm{bin}(\bm c)<\textrm{bin}(\bm d)$.
We see $\bm 0\prec \bm c^\ast_r\prec \bm c^\ast_{r-1}\prec\dots\prec \bm c^\ast_1\prec\bm 1$. 
Let $P$ be a set of elements. 
We assume that any element $\cycle$
satisfies $ \cycle \prec \min P$ 
if $P=\emptyset$ for convenience. 
%
%

For an element $\bm c$,
we define the {\em successor of $\cycle$}
to be the element $\successor(\cycle)$
such that
\begin{align*}
  \textrm{bin}(\successor(\cycle))&=\left\{
  \begin{array}{ll}
    \textrm{bin}(\cycle)+1 & \textrm{if\ }\cycle\notin\{\bm 1,\bm 0\};\\
    0 & \textrm{if\ }\cycle\in\{\bm 1,\bm 0\}.
  \end{array}
  \right.
\end{align*}
%
%
For a positive integer \( a \),
we define \( \successor^{(a)}(\cycle) \coloneqq \successor(\successor^{(a-1)}(\cycle)) \), where we let \( \successor^{(0)}(\cycle) \coloneqq \cycle \).
%
In particular, for $i\in[1,r]$,
$\mathrm{bin}(\successor^{(2^{r-i})}(\cycle))=\mathrm{bin}(\cycle)+\mathrm{bin}(\bm c^\ast_i)$ holds
as long as the overflow does not occur. 
We can compute \( \successor^{(a)}(\cycle) \) in \( \MO(\rank) \) time.

Let $\sol$ be a basis. 
For $x\in\basis$,
we let \( \cut (\basis;x)\coloneqq 
\set{y \in \gft^r\mid (x,y) \text{ is an exchange for } \basis}\cup\{x\} \).
We define the {\em parent} $\pi(\sol)$ of $ \sol $ as follows. 
\[ \pi(\sol) \coloneqq \begin{cases}
    (\sol \setminus \min(\sol \setminus \rtc)) \cup \min (\rtc \cap \cut(\sol; \min(\sol \setminus \rtc))) &\text{if } \sol \ne \rtc,\\
  \emptyset & \text{if } \sol = \rtc. 
\end{cases} \]

By \lemref{exchange}, 
$ \rtc \cap \cut(\sol; \min (\sol \setminus \rtc))$ is not empty
unless $B=\rtc$.
Then $\pi(B)$ is nonempty and a basis if $\sol\ne\rtc$.
%
The following lemma shows that
this definition of the parent induces
a transitive ancestor-descendant relationship between bases. 

\begin{lem}\label{lem:numberofrootcycle}
  For the root $\rtc$,
  let $\sol\ne \rtc $ be a basis.
  It holds that $ |\pi(\sol) \cap \rtc| =  | \sol \cap \rtc| + 1$.
\end{lem}
\begin{proof}
  By the definition of the parent, $ \pi(\sol) $ is obtained by removing an element not in $ \rtc $ and adding an element in $ \rtc $.
\end{proof}

\noindent
It holds that $ |\sol \cap \rtc| = \rank $ if and only if $ \sol = \rtc $.

\subsubsection{Generation of Children}
Let \( \sol \) be a basis.
If $\sol\ne\rtc$,
then the parent $\pi(B)$
is written in the form $\pi(B)=(B\setminus\{\bm y'\})\cup\{\bm x'\}$,
where $\bm x'= \min (\rtc \cap \cut (\child;\min(\child \setminus \rtc)))$
and $\bm y'=\min(\child \setminus \rtc)$. 
This means that, if a basis $P$ is the parent of $B$,
then $B$ can be expressed in the form
$B=(P\setminus\{\bm x\})\cup\{\bm y\}$
for some elements $\bm x\in R$ and $\bm y\notin R\cup P$. 
We show a necessary and sufficient condition that 
$(P\setminus\{\bm x\})\cup\{\bm y\}$  is a child of $ \parent $.


\begin{lem}
  \label{lem:childrennecessary}
  For the root $\rtc$, let \( \child \ne \rtc \) be a basis. 
  Suppose that
  there are elements 
  \( \cast_i \in \rtc \) and \( \add \in B\setminus\rtc \) such that
  \( P = (\child \setminus \set \add)\cup \set {\cycle^\ast_i} \) is a basis.
  Then \( \parent \) is the parent of \( \child \) if and only if 
  both of the following hold:
\begin{description}
\item[(i)] For any $ j \in [i + 1,\rank] $, it holds that $ \det (\bm \parent^{(\cycle^\ast_i,j)})  = 0 $. 
\item[(ii)] It holds that $ \add \prec \min(\parent \setminus \rtc) $.
\end{description}
\end{lem}

\begin{proof}
    Suppose \( \parent=\pi(B) \). 
    Let $ Q\coloneqq (P\setminus\{\bm c^\ast_i\})\cup\{\bm \cycle^\ast_j\}$. 
    The row vectors of $ \bm Q $ are dependent
    since $ \cycle^\ast_i $ is the minimum element in $ \rtc \cap \cut(\child;\min(\child\setminus\rtc)) $.
    Thus, $ \det \bm Q = 0 $ holds. 
    By the cofactor expansion, 
    $ \det \bm Q = \det (\bm \parent^{(\cycle^\ast_i,j)}) = 0 $ holds,
    showing (i).
    The basis $P$ is the parent of $B$, and
    $ \child\setminus\rtc \supsetneq \parent\setminus\rtc $ holds.
    By $\bm y=\min(\child\setminus\rtc)$ and $\bm y\notin P\setminus R$,
    $ \bm y  \prec \min(\parent\setminus\rtc) $ holds for (ii). 

    Suppose that both (i) and (ii) hold.
    For all \( j \in [i+1,\rank] \), 
    it holds that $ \cycle^\ast_j \notin \cut (\child; \add) $ 
    since $ \det (\bm \parent^{(\cycle^\ast_i,j)} ) = \det (\bm \child^{(\add,j)}) = 0 $ by (i).
    Hence, $ \cycle^\ast_i $ is the minimum element in $ \rtc \cap \cut(\child; \add) $.
    We see that $|P\setminus R|=|B\setminus R|-1$.
    By $ \add\notin P\setminus R$, $\add\in B\setminus R$
    and (ii), $ \min(\child \setminus \rtc) = \add $ holds.
    Then we have  $ \pi(\child) = \parent $. 
\end{proof}

In order to generate all $ \remove $ and $ \add $ such that $ \exc{\parent}{\remove}{\add} $ is a child of a basis $ \parent $,
we do the following 1 and 2:
\begin{itemize}
\item[1.] We first choose $ \remove = \cycle^\ast_i \in \parent \cap \rtc $ that satisfies the condition of \lemref{childrennecessary}(i).
  If there is no such $\cycle^\ast_i$, then $ \parent $ has no child. 
\item[2.] For each $\bm c^\ast_i$ in 1,
  we search all $ \add \in  \cut(\parent;\cycle^\ast_i)$
  that satisfy the condition of \lemref{childrennecessary}(ii)
  and employ $\add$ if $\add\notin R$.
  The minimum element in \( \cut(\parent;\cycle^\ast_i) \)
  is $\bm c^\ast_i $
  since we have chosen $\bm c^\ast_i$ so that
  $ \det (\bm \parent^{(\cycle^\ast_i,j)})  = 0 $ holds for all $ j \in [i + 1,\rank] $.  
  Then we enumerate all $ \add \in \cut(\parent;\cycle^\ast_i)$ satisfying that $\bm c^\ast_i\prec \add \prec \min(\parent \setminus \rtc) $ in the increasing order.
\end{itemize}

For 2, 
let \( \bm z \coloneqq \successor(\add) \), \( \bm z' \coloneqq \successor^{(2^{\rank-i})}(\bm z) \) and \( \bm z'' \coloneqq \successor^{(2^{\rank-i+1})}(\bm z) \).
Later in Lemma~\ref{lem:skip2}, 
we show that
at least one of \( \bm z \), \( \bm z' \) and \( \bm z'' \)
is the smallest element larger than \( \add \) in \( \cut(P;\cycle^\ast_i) \), or there is no element larger than \( \add \) in \( \cut(P;\cycle^\ast_i) \).
Hence, it suffices to check only \( \bm z \), \( \bm z' \) and \( \bm z'' \) to find the smallest element larger than \( \add \) in \( \cut(P;\cycle^\ast_i) \).

The following Lemma~\ref{lem:succ} is
used to prove \lemref{skip2}. 
 
\begin{lem}\label{lem:succ}
  For the root $R$, let $B\ne R$ be a basis
  and $P\coloneqq \pi(B)$ denote the parent of $B$
  such that $P=(B\setminus\{\bm y\})\cup\{\bm c^\ast_i\}$
  for ${\bm y}=\min(B-R)$ and $\bm c^\ast_i=\min(R\cap\cut(B;\bm y))$.
  %
  Let $\bm \upsilon$ be an element such that $\bm\upsilon\in\cut(P;\bm c^\ast_i)$ and
  $\bm z\coloneqq\successor(\bm \upsilon)$. 
  If $(P\setminus \{\bm c^\ast_i\})\cup\{\bm z\}$ is not a basis,
  then:
  \begin{description}
  \item[(i)] \( z_j = 0 \) for all integers \( j \in [i+1,\rank] \); and
  \item[(ii)] for every element \( \bm z' \in \gft^\rank \) such that \( z'_p = z_p\) for all \( p \in [1,i] \),
    $(P\setminus \{\bm c^\ast_i\})\cup\{\bm z'\}$ is not a basis. 
  \end{description}
  %
\end{lem}

\begin{proof}
  Let $Q\coloneqq (P\setminus \{\bm c^\ast_i\})\cup\{\bm z\}$.
  We see that 
  \( \det (\bm \parent^{(\cycle^\ast_i,p)})=0 \) for all \( p \in [i+1,\rank] \) holds by \lemref{childrennecessary}(i).
  
  \noindent (i)
    Suppose that there exists $j\in[i+1,r]$ such that \( z_j=1 \).
    It holds that \( z_p = \upsilon_p\) for all \( p \in [1,i] \subseteq [1,j-1] \) by the definition of successor.
    We see 
    \begin{align*}
      \det\bm Q&= \sum_{p=1}^\rank z_p \det (\bm \parent^{(\cycle^\ast_i,p)}) = \sum_{p=1}^i z_p \det (\bm \parent^{(\cycle^\ast_i,p)}) \\
      &= \sum_{p=1}^i \upsilon_p \det (\bm \parent^{(\cycle^\ast_i,p)}) = \sum_{p=1}^\rank \upsilon_p \det (\bm \parent^{(\cycle^\ast_i,p)})
      = \det\bm B=1,
      \end{align*}
    showing that $Q$ is a basis.

    \noindent (ii)
    Let $Q'\coloneqq (P\setminus \{\bm c^\ast_i\})\cup\{\bm z'\}$. Similarly to (i), we see
    \begin{align*}
      \det \bm Q'&=\sum_{p=1}^\rank z'_p \det (\bm \parent^{(\cycle^\ast_i,p)}) = \sum_{p=1}^i z'_p \det (\bm \parent^{(\cycle^\ast_i,p)})\\
      &= \sum_{p=1}^i z_p \det (\bm \parent^{(\cycle^\ast_i,p)}) = \sum_{p=1}^\rank z_p \det (\bm \parent^{(\cycle^\ast_i,p)}) =\det \bm Q= 0.
    \end{align*}
    \hfill
\end{proof}

\begin{lem}\label{lem:skip2}
  For the root $R$, let $B\ne R$ be a basis
  and $P\coloneqq \pi(B)$ denote the parent of $B$
  such that $P=(B\setminus\{\bm y\})\cup\{\bm c^\ast_i\}$
  for ${\bm y}=\min(B-R)$ and $\bm c^\ast_i=\min(R\cap\cut(B;\bm y))$.
  Let $\bm\upsilon$ be an element such that $\bm\upsilon\in\cut(P;\bm c^\ast_i)$,
  \( \bm z \coloneqq \successor(\bm\upsilon) \), 
  \( \bm z' \coloneqq \successor^{(2^{\rank-i})}(\bm z) \) and
  \( \bm z'' \coloneqq \successor^{(2^{\rank-i+1})}(\bm z) \).
  \invis{
    For a basis \( \sol \ne \rtc \), 
    let \( P = \pi(\sol) = (\child \setminus \set \add)\cup \set {\cycle^\ast_i} \) denote the parent of \( \sol \) 
    and \( k \) denote the largest index 
    such that \( \det (\bm \parent^{(\cycle^\ast_i,k)}) = 1 \), 
    where \( \add = \min(\sol \setminus \rtc) \) 
    and \( \cycle^\ast_i \in \rtc \).
    Let \( \bm z \coloneqq \successor(\add) \), 
    \( \bm z' \coloneqq \successor^{(2^{\rank-k})}(\bm z) \) and
    \( \bm z'' \coloneqq \successor^{(2^{\rank-k+1})}(\bm z) \).
  }
  \begin{description}
    \item[(i)] If \( \bm z'' \ne \bm 0 \), then at least one of \( \bm z \), \( \bm z' \) and \( \bm z'' \) belongs to \( \cut(\parent;\cycle^\ast_i) \).
    \end{description}
    From {\rm(ii)} to {\rm(iv)},
    we assume that $\bm z\ne\bm 0$ and
    \( \bm z \notin \cut(\parent;\cycle^\ast_i) \). 
    \begin{description}
    \item[(ii)] If \( \bm z' \ne \bm 0 \), then
      no element \( \cycle \)
      such that \( \bm z \prec \bm c \prec \bm z' \)
      belongs to \( \cut(\parent;\cycle^\ast_i) \).
    \item[(iii)] If \( \bm z'' \ne \bm 0 \)
      and \( \bm z' \notin \cut(\parent;\cycle^\ast_i) \),
      then no element \( \cycle \) such that
      \( \bm z \prec \bm c \prec \bm z'' \) belongs
      to \( \cut(\parent;\cycle^\ast_i) \).
    \item[(iv)] If \( \bm z'' = \bm 0 \) and \( \bm z' \notin \cut(\parent;\cycle^\ast_i) \), then
      no element $\cycle$ such that \( \cycle \succ \bm z \)
      belongs to \( \cut(\parent;\cycle^\ast_i) \).
    \end{description}
\end{lem}

\begin{proof}
  Let $Q\coloneqq (P\setminus\{\bm c^\ast_i\})\cup\{\bm z\}$,
  $Q'\coloneqq (P\setminus\{\bm c^\ast_i\})\cup\{\bm z'\}$ and
  $Q''\coloneqq (P\setminus\{\bm c^\ast_i\})\cup\{\bm z''\}$,
  respectively.

  \noindent
  (i)
  By the definition of the successor,
  if \( z'_i = 1 \), then \( z'_p = z_p \) for all \( p \in [1,\rank] \setminus \set i \) and \( z_i = 0 \) holds; 
  and if \( z'_i = 0 \), then \( z'_p = z''_p \) for all \( p \in [1,\rank] \setminus \set i \) and \( z''_i = 1 \) holds
  by \( \bm z'' \ne \bm 0 \).
  The set $P$ is a basis and 
  $\det \bm P^{(\bm c^\ast_i,i)}=1$.
  Then either $\det Q\ne\det Q'$ or $\det Q'\ne\det Q''$ holds. 

  \noindent
  (ii) The set $Q$ is not a basis.
  By Lemma~\ref{lem:succ}(i),
  $z_p=0$ holds for all $p\in[i+1,r]$. 
  By Lemma~\ref{lem:succ}(ii),
  any element $\bm c$
  such that $c_p=z_p$ for all $p\in[1,i]$
  is not a basis, and such $\bm c$ satisfies
  $\bm z\prec\bm c\prec\bm z'$. 

  \noindent
  (iii) 
  The set $Q'$ is not a basis. 
  By Lemma~\ref{lem:succ}(i),
  $z'_p=0$ holds for all $p\in[i+1,r]$. 
  By Lemma~\ref{lem:succ}(ii),
  any element $\bm c$
  such that $c_p=z'_p$ for all $p\in[1,i]$
  is not a basis, and such $\bm c$ satisfies
  $\bm z'\prec\bm c\prec\bm z''$.
  With (ii) and
  \( \bm z' \notin \cut(\parent;\cycle^\ast_i) \),
  no $\cycle$ such that 
  \( \bm z \prec \bm c \prec \bm z'' \) belongs
  to \( \cut(\parent;\cycle^\ast_i) \).
  

  \noindent
  (iv)
  If \( \bm z' \ne \bm 0 \), then no element \( \cycle \)
  such that \( \bm z \prec \cycle \prec \bm z' \)
  belongs to \( \cut(\parent;\cycle^\ast_i) \) by (ii).
  All elements $\bm c'$ such that $\bm c'\succ\bm z'$
  satisfy $c'_p=z'_p=1$ for all $p\in[1,i]$
  since $\bm z'\ne\bm 0$ and $\bm z''=\bm0$.
  By \lemref{succ}(ii), $\bm c'\notin\cut(\parent;\cycle^\ast_i)$
  holds.
  If $\bm z'=\bm 0$, then 
  all elements $\bm c$ such that $\bm c\succ\bm z$
  satisfy $c_p=z_p=1$ for all $p\in[1,i]$
  since $\bm z\ne\bm 0$ and $\bm z'=\bm0$.
  By \lemref{succ}(ii), $\bm c\notin\cut(\parent;\cycle^\ast_i)$
  holds.
\end{proof}

\subsubsection{Algorithm for Enumerating Bases}
For a basis $P$ and $\bm c^\ast_i\in P$ that satisfies the condition
of \lemref{childrennecessary}(i), let $\bm \upsilon\in\cut(\parent;\cycle^\ast_i)$
such that $\bm \upsilon\prec\min(P\setminus R)$. 
Using \lemref{skip2}, we generate the smallest element in $\cut(\parent;\cycle^\ast_i)$
that is strictly larger than $\bm \upsilon$ and smaller than $\min(P\setminus R)$
or conclude that no such element exists as follows.
If $\bm z\coloneqq\successor(\bm \upsilon)$ equals to $\bm 0$,
then no solution exists. Otherwise,
if $\bm z\in\cut(\parent;\cycle^\ast_i)$ (i.e., $\det((P\setminus\{\cycle^\ast_i\})\cup\{\bm z\})=1$), then $\bm z$ is the solution.
To consider the remaining cases (i.e., $\bm z\ne\bm 0$ and
$\bm z\not\in\cut(\parent;\cycle^\ast_i)$), 
let \( \bm z' \coloneqq \successor^{(2^{\rank-i})}(\bm z) \)
and \( \bm z'' \coloneqq \successor^{(2^{\rank-i+1})}(\bm z) \).
\begin{itemize}
\item If $\bm z'=\bm 0$, then $\bm z''=\bm 0$
  and $\bm z'\notin\cut(\parent;\cycle^\ast_i)$ hold.
  We can conclude that no solution exists by (iv).
\item If $\bm z'\ne\bm 0$, then $\bm z'$ is the smallest candidate by (ii).
  Hence, if $\bm z'\in\cut(\parent;\cycle^\ast_i)$, then we are done with $\bm z'$,
  and otherwise:
  \begin{itemize}
  \item if $\bm z''=\bm 0$, then no solution exists by (iv); and
  \item otherwise, $\bm z''$ is the solution by (i) and (iii). 
  \end{itemize}
\end{itemize}
For a solution $\bm \upsilon'$ that is obtained by this algorithm,
if $\bm \upsilon'\notin R$, then $(P-\{\bm c^\ast_i\})\cup\{\bm \upsilon'\}$ is a child of $P$. 
We summarize the algorithm to generate all children
of a given basis in Algorithm~\ref{alg:children}.

\begin{algorithm}[t!]
  \caption {An algorithm to enumerate all children of a given basis}
  \label{alg:children}
  \begin{algorithmic}[1]
\Require {A basis $\parent$ (i.e., a set of linearly independent $r$-dimensional $r$ vectors over $\gft$)} 
\Ensure{All children of $ \parent $}
\algfor{each $ \cycle^\ast_i \in \parent \cap \rtc $}{
  \Comment{$\rtc$ denotes the root basis.}
      \algif{$\det (\bm \parent^{(\cycle^\ast_i,j)} ) = 0 $ for all $j \in [i + 1,\rank]$ }{
          \State $ \bm\upsilon \coloneqq \cycle_i^\ast $;
          \algwhile{$ \bm\upsilon \prec \min(\parent \setminus \rtc) $ and \( \bm\upsilon \ne \bm 0 \)}{
              \algif{$ \bm\upsilon \notin \rtc $}{
                  \State \textbf{output} $ (P\setminus \{ \cycle^\ast_i\}) \cup \{\bm\upsilon\} $
              };
              \State \( \bm z \coloneqq \successor(\bm\upsilon) \); $\bm \upsilon\coloneqq\bm z$;
              \algif{$\bm z\ne\bm 0$ and $\bm z\notin\cut(\parent;\cycle^\ast_i)$}{ 
                  \State \( \bm z' \coloneqq \successor^{(2^{\rank-i})}(\bm z) \); $\bm \upsilon\coloneqq \bm z'$;
                  \algif{$\bm z'\ne\bm 0$ and $\bm z'\notin\cut(\parent;\cycle^\ast_i)$}{ 
                      \State \( \bm z'' \coloneqq \successor^{(2^{\rank-i+1})}(\bm z) \); $\bm \upsilon\coloneqq \bm z''$
                  }
              }
          }
      }
  }
\end{algorithmic}
\end{algorithm}

\begin{lem}
  \label{lem:children}
  Given a basis $P$,
  Algorithm~\ref{alg:children} outputs all children of $P$
  in $ \MO(r^3+\rank^2 \tauDET(\rank)) $ delay.  
\end{lem}
\begin{proof}
  The correctness follows by the discussion before this lemma.
  We analyze the time complexity. 
  The outer for-loop is repeated $\MO(r)$ times.
  Within the for-loop, we compute determinants several times,
  and each can be obtained in $\MO(r)$ time
  by computing  $\det(\bm P^{(\bm c^\ast_i,j)})$ for $j\in[1,r]$ by preprocessing,
  which can be done in $\MO(r\tauDET(r))$ time. 
  All operations in the while-loop can be done in $\MO(r)$ time,
  and at least one solution is output within $r+1$ iterations
  since $|R|=r$. 
  We see that the delay is $\MO(r(r\tauDET(r)+r^2))=\MO(r^3+r^2\tauDET(r))$. 
  %
\end{proof}


Based on Algorithm~\ref{alg:children}, 
we present a reverse search algorithm
that enumerates all $r\times r$ bases in Algorithm~\ref{alg:allbases}.

\begin{algorithm}[t!]
\caption{An algorithm to enumerate all bases}
\label{alg:allbases}
\begin{algorithmic}[1]
  \Require{A positive integer $r$} 
  \Ensure{All bases (i.e., all sets of linearly independent $r$-dimensional $r$ vectors over $\gft$)}
\State $R\coloneqq\{\bm c^\ast_1,\bm c^\ast_2,\dots,\bm c^\ast_r\}$; 
\State call \textsc{EnumAllBases}\( (R,0) \)
\Statex
\procedure{EnumAllBases}{$P,d$}{
  \algif{\( d \) is even}{
    \State output $P$
  };
  \algfor{each $ \cycle^\ast_i \in \parent \cap \rtc $}{
    \algif{$\det (\bm \parent^{(\cycle^\ast_i,j)} ) = 0 $ for all $j \in [i + 1,\rank]$ }{
      \State $ \bm\upsilon \coloneqq \cycle_i^\ast $;
      \algwhile{$ \bm\upsilon \prec \min(\parent \setminus \rtc) $ and \( \bm\upsilon \ne \bm 0 \)}{
        \algif{$ \bm\upsilon \notin \rtc $}{
          \State call \textsc{EnumAllBases}$((P\setminus\{\bm c^\ast_i\})\cup\{\bm \upsilon\},d+1)$
        };
        \State \( \bm z \coloneqq \successor(\bm\upsilon) \); $\bm \upsilon\coloneqq\bm z$;
        \algif{$\bm z\ne\bm 0$ and $\bm z\notin\cut(\parent;\cycle^\ast_i)$}{
          \State \( \bm z' \coloneqq \successor^{(2^{\rank-i})}(\bm z) \); $\bm \upsilon\coloneqq \bm z'$;
          \algif{$\bm z'\ne\bm 0$ and $\bm z'\notin\cut(\parent;\cycle^\ast_i)$}{
            \State \( \bm z'' \coloneqq \successor^{(2^{\rank-i+1})}(\bm z) \); $\bm \upsilon\coloneqq \bm z''$
          }
        }
      }
    }
  };
  \algif{\( d \) is odd}{
    \State output $P$
  }
}
\end{algorithmic}
\end{algorithm}

\begin{thm}
  Given a positive integer $r$,
  Algorithm~\ref{alg:allbases} enumerates all $r\times r$ bases
  in $\MO(r^3+r^2\tauDET(r))$ delay. 
\end{thm}
\begin{proof}
  The correctness is due to \lemref{children}.
  When \textsc{EnumAllBases}$(P,d)$ is called, 
  we change the timing when we output $P$
  by the parity of $d$, that is the depth of the search tree. 
  This is called {\em alternative output method}~\cite{alternative},
  and it is known that the delay is bounded by the processing time
  of each node of the search tree. 
\end{proof}

\assumref{binary} includes binary matroids
from cycle space and cut space~\cite{cycleandcutspace}.\footnote{For cut space, there is an exception such that the sum is not a cut if it is $\bm 0$.} 
For the cycle space,
the rank is $r=\MO(m)$ and 
we can construct a basis from a spanning tree~\cite{mcb}
in $\MO(\rank(n+m))$ time.
For the cut space, 
the rank is $r=\MO(n)$ and 
the family of \( n-1 \) cuts \( E(\set v; V \setminus \set v) \) for \( n-1 \) vertices in \( V \) is a basis~\cite{cycleandcutspace},
which can be constructed in \( \MO(\rank(n+m)) \) time.
The $r$ elements in the obtained basis correspond to
the root $R=\{\cycle^\ast_1,\cycle^\ast_2,\dots,\cycle^\ast_r\}$.
Using $\tauDET(r)=\MO(r^3)$, we have the following corollary. 

\begin{cor}
  Given a connected graph $G$,
  we can enumerate all cycle bases 
  in $\MO(m^5)$ delay and
  and all cut bases in $\MO(n^5)$ delay. 
\end{cor}

\section{Concluding Remarks}
\label{sec:conc}
In this paper,
we have studied problems of enumerating
bases in a matroid with an exponentially large ground set.
Examples of such a matroid includes
binary matroids from cycle space, path space and cut space.
The first result is an incremental polynomial algorithm
that enumerates all minimum-weight bases.
We have shown that it is applicable to
the above binary matroids so that the time complexity
is polynomially bounded by the graph size.
The second result is a poly-delay algorithm
that enumerates all bases, which
immediately yields poly-delay (with respect to graph size)
enumeration of bases in binary matroids from cycle space and cut space.

For future work, 
it would be interesting to explore
whether or not it is possible to enumerate all minimum bases
in poly-delay in our context (i.e., 
the bounding polynomial does not contain the ground set size)
and/or in poly-space.
Among the remaining issues are to study the $k$-best version,
to improve the complexities, to look for other applications, and so on.

\clearpage

\end{document}